\UseRawInputEncoding\pdfoutput=1
\documentclass[12pt]{article}
\usepackage{color, fullpage}
\usepackage{authblk,amsmath,amsthm,amsfonts,amssymb,amscd,enumerate,graphicx,color,appendix,tabularx,cite}
\allowdisplaybreaks[4]

\newtheorem {lemma}{Lemma}[section]
\newtheorem {theorem}{Theorem}[section]

\newtheorem {corollary}{Corollary}[section]

\allowdisplaybreaks [4]

\begin{document}

\title{\bf On the $\alpha$-spectral radius of hypergraphs}
\author{Haiyan Guo\footnote{E-mail: ghaiyan0705@m.scnu.edu.cn},  Bo Zhou\footnote{Corresponding author. E-mail:
zhoubo@scnu.edu.cn}, Bizhu Lin\footnote{E-mail: 2270927176@qq.com}\\
School of Mathematical sciences, South China Normal University, \\
Guangzhou 510631, P. R. China }

\date{}
\maketitle

\begin{abstract}  For real $\alpha\in [0,1)$ and  a hypergraph $G$, the $\alpha$-spectral radius of $G$ is the largest eigenvalue of  the matrix $A_{\alpha}(G)=\alpha D(G)+(1-\alpha)A(G)$, where  $A(G)$ is the adjacency matrix of $G$, which is a symmetric matrix with zero diagonal such that for distinct vertices $u,v$ of $G$, the $(u,v)$-entry of $A(G)$ is exactly the number of edges containing both $u$ and $v$, and $D(G)$ is the diagonal matrix of row sums of $A(G)$. We study the $\alpha$-spectral radius of a hypergraph that is uniform or not necessarily uniform. We propose some local grafting operations that increase or decrease the $\alpha$-spectral radius of a hypergraph. We determine the unique hypergraphs with maximum $\alpha$-spectral radius among $k$-uniform hypertrees, among $k$-uniform  unicyclic hypergraphs, and among $k$-uniform hypergraphs with fixed number of pendant edges.
We also determine the unique hypertrees with maximum $\alpha$-spectral radius among hypertrees with given number of vertices and edges, the unique hypertrees with the first three largest (two smallest, respectively)  $\alpha$-spectral radii among hypertrees with given number of vertices, the unique hypertrees with minimum $\alpha$-spectral radius among the hypertrees that are not $2$-uniform,  the unique hypergraphs with the first two largest (smallest, respectively) $\alpha$-spectral radii among unicyclic hypergraphs with given number of vertices, and the unique hypergraphs with maximum $\alpha$-spectral radius among hypergraphs with fixed number of pendant edges.\\ \\
{\bf AMS classifications:} 05C65,  05C50, 15A18\\  
{\bf Key words:}  hypergraph,  adjacency matrix, $\alpha$-spectral radius, local grafting operation
\end{abstract}

\section{Introduction}

A hypergraph $G=(V, E)$ consists of vertex set $V(G)=V$ and edge set $E(G)=E$, where $e\subseteq V(G)$ and $|e|\ge 2$ for each $e\in E(G)$. For integer $k\ge 2$, a hypergraph $G$ is $k$-uniform if $|e|=k$ for each $e\in E(G)$. A $2$-uniform hypergraph is just a (simple) graph.
The concept of a hypergraph may be viewed as  a variant of the concept of  a block design. A block design is a $k$-uniform hypergraph for some $k$,  with objects as vertices and blocks as edges, so that any realization of a hypergraph associated with a block design is simultaneously a realization of that design \cite{CaYu,Linek}.

For $u,  v\in V(G)$, if they are contained in some edge of $G$, then we say that they are adjacent or $v$ is a neighbor of $u$, written as $u\sim v$. For $u\in V(G)$, let $N_{G}(u)$ be the set of neighbors of $u$ in $G$ and $E_{G}(u)$ be the set of edges containing $u$ in $G$. The degree of a vertex $u$ in $G$, denoted by $d_{G}(u)$ or $d(u)$, is the number of edges of $G$ which contain $u$.
A path in a hypergraph $G$ is a vertex-edge alternating sequence  $v_0e_1v_1\dots e_pv_p$ such that  $v_0, \dots, v_p$ are distinct vertices, $e_1, \dots, e_p$ are distinct edges, and $v_{i-1},  v_i\in e_i$ for $i=1,  \dots,  p$; such a path is also called a path from $v_0$ to $v_p$.
A cycle in a hypergraph $G$ is a vertex-edge alternating sequence  $v_0e_1v_1\dots e_pv_0$ such that  $v_0, \dots, v_{p-1}$ are distinct vertices, $e_1, \dots, e_p$ are distinct edges, and $v_{i-1},  v_i\in e_i$ for $i=1,  \dots,  p$ (with $v_p$ interpreted as $v_0$). The number of edges in a path or a cycle is known as its length.
If there is a path from $u$ to $v$ for any $u, v\in V(G)$, then we say that $G$ is connected. A hypergraph is said to be linear if any two edges have at most one vertex in common. A hypertree is a connected hypergraph with no cycles. Obviously, hypertrees are linear. A unicyclic hypergraph is a connected hypergraph with exactly one cycle and the length of the unique cycle is known as its girth.


A hyperstar is a hypertree in which all edges contain a vertex in common, and this vertex is called its center. If a hyperstar has exactly one edge, then we may choose any vertex as its center.   The $k$-uniform hyperstar on $n$ vertices is denoted by $S_{n,k}$. Let $S_n=S_{n,2}$, which is the ordinary star (graph) on $n$ vertices.

For a $k$-uniform hypertree $T$ with $V(T)=\{ v_{1}, \dots, v_{n} \}$, if $E(T)=\{ e_{1}, \dots, e_{m} \}$ with $m=\frac{n-1}{k-1}$, where $e_i=\{v_{(i-1)(k-1)+1}, \ldots, v_{(i-1)(k-1)+k}\}$ for $i={1, 2, \ldots, m}$, then we call $T$ a $k$-uniform loose path, denoted by $P_{n, k}$. Let $P_n=P_{n,2}$, which is the ordinary path (graph) on $n$ vertices.

For a hypergraph $G$ and $e=\{ v_{1}, \dots, v_{k} \}\in E(G)$, if $d_{G}(v_{i})=1$ for $1\leq i\leq k-1$ and $d_{G}(v_{k})\geq 2$, then we call $e$ a pendant edge (at $v_{k}$).

Let $G$ be a hypergraph on $n$ vertices. The adjacency matrix of $G$ is the $n\times n$ matrix $A(G)=(a_{uv}(G))_{u,v\in V(G)}$ with
\[
a_{uv}(G)=\begin{cases}
|\{e\in E(G):\{u,v\}\subseteq e \}| & \mbox{if $u\neq v$},\\[1mm]
0 & \mbox{if $u=v$}.
\end{cases}
\]
The spectral radius of $G$ is the largest eigenvalue of $A(G)$.
Feng and Li \cite{FL} studied estimates for eigenvalues of the adjacency matrix of a regular $k$-uniform hypergraph.
Li and Sol\'{e} \cite{LS} established some further properties for eigenvalues of the adjacency matrix of a regular $k$-uniform hypergraph.  Mart\'inez et al. \cite{Ma} studied the spectra of the adjacency matrices of hypergraphs associated to ``homogeneous spaces'' of $GL_n$   over the finite field $\mathbb{F}_q$, and obtained
an infinite family of Ramanujan hypergraphs.
Lin and Zhou \cite{LZ} considered the extremal problems on the the spectral radius of  a hypergraph.
Among others, they 
determined the unique $k$-uniform hypertrees with the first three largest  spectral radii,  and the unique $k$-uniform unicyclic hypergraphs ($k$-uniform linear unicyclic hypergraphs, respectively) with the first two largest spectral radii, and the unique  hypergraphs with maximum spectral radius among connected
$k$-uniform hypergraphs with  fixed number of pendant edges.

For a hypergraph $G$ with $u\in V(G)$, let $r_G(u)=\sum_{v\in V(G)}a_{uv}(G)$, which is the row sum of $A(G)$ at $u$. Note that $r_G(u)=\sum_{e\in E_G(u)}(|e|-1)$ for $u\in V(G)$.
Let $D(G)$ be the diagonal matrix of row sums of $A(G)$. Let $Deg(G)$ be the diagonal matrix of vertex degrees of $G$. If $G$ is $k$-uniform, then $r_G(u)=(k-1)d_G(u)$, and thus  $D(G)=(k-1)Deg(G)$.  The Laplacian matrix of a hypergraph $G$ is defined as $L(G)=D(G)-A(G)$, see \cite{Ro1}.
In \cite{Ro2}, Rodr\'iguez used the eigenvalues of the Laplacian matrix of a hypergraph to bound its bipartition width, averaged minimal cut, isoperimetric number, max-cut, independence number and domination number. More results along this line may be found in \cite{Ro1,Ro3}.
The matrix $Q(G)=D(G)+A(G)$ is known as the signless Laplacian matrix of $G$ \cite{LZ}. The signless Laplacian spectral radius of $G$  is the largest eigenvalue of $Q(G)$.

Let $\alpha$ be a real number in $[0,1)$. A matrix  $A_{\alpha}(G)$ for  a hypergraph $G$ is defined as $A_{\alpha}(G)=\alpha D(G)+(1-\alpha)A(G)$.  The $\alpha$-eigenvalues of a hypergraph  $G$ are the eigenvalues of $A_{\alpha}(G)$, and the largest eigenvalue of $A_{\alpha}(G)$ is called
the $\alpha$-spectral radius of $G$, denoted by $\rho_{\alpha}(G)$.  For an ordinary graph $G$, $D(G)=Deg (G)$, and this matrix $A_{\alpha}(G)$ was proposed by Nikiforov \cite{Ni1} to study the spectral properties of the adjacency matrix $A_0(G)$ and the signless Laplacian matrix $2A_{1/2}(G)$ in a unified way.
Early, the $0$-spectral radius has been extensively studied, see \cite{BrH,CR,LZ}, while the $\frac{1}{2}$-spectral radius has also received much attention, see \cite{HL,Odede,LZ}.

We mention that a much different treat via adjacency tensors (hypermatrices) may be found in \cite{GZ,LGZ}. However,  the notation of the adjacency tensor  does not have any immediate relationship with  the spectral radius of a hypergraph via its adjacency matrix.


In this paper, we study the $\alpha$-spectral radius of a hypergraph  that is uniform or not necessarily uniform.
We propose some local operations (grafting operations) that increase or decrease the $\alpha$-spectral radius of a hypergraph.
We determine the unique hypergraphs with maximum $\alpha$-spectral radius among $k$-uniform hypertrees, among $k$-uniform  unicyclic hypergraphs, and among $k$-uniform hypergraphs with fixed number of pendant edges.
We determine the unique hypertrees with  maximum $\alpha$-spectral radius among hypertrees with given number of vertices and edges. We also determine the unique hypertrees with the first three largest (two smallest, respectively) $\alpha$-spectral radii among hypertrees with given number of vertices.  We determine the unique hypertrees with minimum $\alpha$-spectral radius among the hypertrees  that are not $2$-uniform with given number of vertices.  We  determine the unique hypergraphs with the first two largest (smallest, respectively) $\alpha$-spectral radii among unicyclic hypergraphs with given number of vertices. Finally, we determine the unique hypergraphs with maximum $\alpha$-spectral radius among hypergraphs with fixed number of pendant edges.

\section{Preliminaries}

In the rest of this paper we assume that $\alpha\in [0,1)$. Let $G$ be a hypergraph.
 Note that $A_{\alpha}(G)$ is a symmetric nonnegative matrix. If $G$ is connected, then $A_{\alpha}(G)$ is irreducible, and thus by Perron-Frobenius theorem, there is a unique positive unit eigenvector corresponding to $\rho_{\alpha}(G)$, which is called the $\alpha$-Perron vector of $G$, denoted by $x(G)$.

Let $V(G)=\{ v_{1}, \dots, v_{n} \}$ and $x=( x_{v_{1}}, \dots, x_{v_{n}} )^{\top}\in \mathbb {R}^n$. Then
\[
x^\top A_{\alpha}(G) x=\sum_{\{u,v\}\subseteq V(G)} a_{uv}(G)\left(\alpha(x_u^2+x_v^2)+2(1-\alpha)x_u x_v\right).
\]
For $x=x(G)$ and each $u\in V(G)$,
\[
\rho_{\alpha}(G)x_u=\sum_{v\in N_{G}(u)}a_{uv}(G)(\alpha x_u+(1-\alpha)x_v),
\]
which is called the eigenequation of $G$ at $u$.
For a hypergraph $G$ on $n$ vertices and a unit column vector $x\in \mathbb {R}^n$ with at least one nonnegative entry, by Rayleigh's principle, we have $\rho_{\alpha}(G)\geq x^\top A_{\alpha}(G) x$ with equality if and only if $x=x(G)$.

For a hypergraph $G$ with $V_0\subseteq V(G)$, let $G[V_0]$ be the subhypergraph of $G$ induced by $V_0$.

Let $G$ be a hypergraph with $u,v\in V(G)$ and $e_{1}, \dots, e_{r}\in E(G)$ such that $u\notin e_{i}$ and $v\in e_i$ for $1\leq i\leq r$. Let $e'_i=(e_{i}\setminus \{v\}) \cup \{u\}$ for $1\leq i\leq r$. Suppose that $e'_i \notin E(G)$ for $1\leq i\leq r$. Let $G'$ be the hypergraph with $V(G')=V(G)$ and $E(G')=(E(G)\setminus \{e_{1}, \dots, e_{r}\}) \cup \{e'_{1}, \dots, e'_{r}\}$. Then we say that $G'$ is obtained from $G$ by moving edges $\{e_{1}, \dots, e_{r}\}$ from $v$ to $u$.

Let $G$ be a hypergraph with $e_{1},e_{2}\in E(G)$ and $u\in V(G)$ such that $u\notin e_{1}$ and $u\in e_{2}$, where $|e_{2}|\geq 3$. Let $e'_1=e_{1} \cup \{u\}$ and $e'_2=e_{2}\setminus \{u\}$. Suppose that $e'_1 ,e'_2\notin E(G)$. Let $G'$ be the hypergraph with $V(G')=V(G)$ and $E(G')=(E(G)\setminus \{e_{1},  e_{2}\}) \cup \{e'_{1},  e'_{2}\}$. Then we say that $G'$ is obtained from $G$ by moving vertex $u$ from $e_{2}$ to $e_{1}$.

For $k\geq 3$, let $e=\{w_{1}, \dots, w_{k}\}$ be an edge of a hypergraph $G$. Let $e_{1}=\{w_{1}, w_{2}\}$ and $e_2=e\setminus \{w_{2}\}$. Suppose that $e_1 ,e_2\notin E(G)$. Let $G'$ be the hypergraph with $V(G')=V(G)$ and $E(G')=(E(G)\setminus \{e\}) \cup \{e_{1},  e_{2}\}$. Then we say that $G'$ is obtained from $G$ by removing vertex $w_{2}$ from $e$ and attaching an edge $\{w_{1}, w_{2}\}$ to $w_{1}$.

For a hypergraph $G$ with $u\in V(G)$, we form a new hypergraph $G'$ by adding a new vertex $w$ and a new edge $\{u,w\}$. In this case, we say $G'$ is obtained from $G$ by attaching a pendant vertex $w$ to $u$.

\begin{lemma} \label{Lem21}
Let $G$ be a connected hypergraph with $\varphi$ being an automorphism of $G$, and $x=x(G)$. If   $\varphi (u)=v$ for $u,v\in V(G)$, then $x_{u}= x_{v}$.
\end{lemma}

\begin{proof}
Let $P=(p_{uv})_{u,v\in V(G)}$ be the permutation matrix such that  $p_{uv}=1$ if and only if  $\varphi (u)=v$ for  $u,v\in V(G)$.
Then $A_{\alpha}(G)=PA_{\alpha}(G) P^{\top}$.  Thus  $\rho_{\alpha}(G)=x^\top A_{\alpha}(G)x=(P^\top x)^\top A_{\alpha}(G)(P^\top x)$. As $P^\top x$ is a positive unit vector, $P^\top x$ is also the  $\alpha$-Perron vector of $G$, so $P^\top x = x$.
\end{proof}

For $s\times t$ matrices  $B=(b_{ij})$ and $C=(c_{ij})$, if  $b_{ij}\leq c_{ij}$ for $1\le i\le s$ and $1\le j\le t$, and $B\ne C$, then we write $B<C$.
Let $G$ and $H$ be two hypergraphs with $n$ vertices. It is obvious that $A(G)<A(H)$ if and only if  $A_{\alpha}(G)<A_{\alpha}(H)$.

Let $\lambda_1(B)$ be the spectral radius (i.e., largest modulus of  eigenvalues) of a nonnegative square matrix $B$.

\begin{lemma} \label{Lem22}\cite{Mi}
Let $B=(b_{ij})$ and $C=(c_{ij})$ be  nonnegative matrices of order $n$. Suppose that $C$ is irreducible and $B<C$. Then $\lambda_1(B)< \lambda_1(C)$.
\end{lemma}

Let $G$ be a connected $k$-uniform hypergraph, and $e$ a subset of $V(G)$ with $|e|=k$ and $e\notin E(G)$. Then by Lemma \ref{Lem22},
$\rho_{\alpha}(G) <\rho_{\alpha}(G+e)$.
Particularly, if $G$ is a connected graph and $H$ is a proper subgraph of $G$, then by Lemma \ref{Lem22}, $\rho_{\alpha}(H)<\rho_{\alpha}(G)$.

We need the following
well known lemma, see, e.g., \cite[Theorem 2.2.1(iv)]{AE}.

\begin{lemma}\label{comp}
Let $M$ be an $n\times n$ irreducible nonnegative  matrix and $y$ a column vector of dimension $n$. If $y>0$ and $My<\rho y$, then $\lambda_1(M)<\rho$.
\end{lemma}

For a graph $G$ with $u\in V(G)$ and a positive integer $k$,
let $G'$ be the graph  obtained from $G$ and a path $P_k$ by adding an edge connecting $u$ and a terminal vertex of the path $P_k$.  In this case, we say $G'$ is obtained from $G$ by attaching a path of length $k$ at $u$ and we denote such a graph $G'$ by
 $G(u;k)$. Let   $G(u;0)=G$. The following two lemmas proven in \cite{GuoZ} were established in  \cite{LF} for $\alpha=0$, and they appeared to be two conjectures in \cite{Ni1}.

\begin{lemma} \label{Lem24} \cite{GuoZ}
Let $u$ be a vertex of a nontrivial connected graph $G$, and let $G_{k,\ell}=G(u; k)(u;\ell)$, where $k\geq \ell\geq 1$. Then $\rho_{\alpha}(G_{k,\ell})>\rho_{\alpha}(G_{k+1,\ell-1})$.
\end{lemma}

\begin{lemma} \label{Lem25} \cite{GuoZ}
Let $u,v$ be adjacent vertices of a connected graph $G$ with degrees at least $2$, and let $G^{k,\ell}=G(u;k)(v; \ell)$,
where $k\geq \ell\geq 1$. Then $\rho_{\alpha}\left(G^{k,\ell}\right)>\rho_{\alpha}\left(G^{k+1,\ell-1}\right)$.
\end{lemma}

An internal path in a graph $G$ is of one of two types:

$(i)$ A sequence of vertices $v_0 v_1\dots v_{k+1}$ ($k\ge 2$), where $v_0, v_1,\dots, v_{k}$ are distinct,  $v_{i-1}$ and $v_i$ are adjacent for $i=1, \dots, k+1$, $v_{k+1}=v_{0}$, $d_G(v_0)\geq 3$, and  $d_G(v_i)=2$ for $i=1,2,\dots,k$;

$(ii)$ A sequence of distinct vertices $v_0 v_1\dots  v_{k+1}$ ($k\ge 0$) such that $v_{i-1}$ and $v_i$ are adjacent for $i=1, \dots, k+1$,  where $d_G(v_0)\geq 3$, $d_G(v_{k+1})\geq 3$, and $d_G(v_i)=2$ whenever $1\le i\le k$.

For $n\ge 5$, let $W_n$ be the tree obtained from the path $P_{n-4}$ by attaching two pendant vertices to both terminal vertices, respectively. Obviously, $W_5\cong S_5$.

For $n\ge 4$, let $Z_n$ be the tree obtained from the path $P_{n-2}$ by attaching two pendant vertices to one terminal vertex. For $n\ge 6$, $Z'_n$ be the tree obtained from the path $P_{n-1}=v_1v_2 \dots v_{n-1}$ by attaching a pendant vertex $v_n$ to $v_{n-3}$.

All connected graphs with $0$-spectral radius at most $2$ have been determined in \cite{Smi}. From this, connected graphs with $0$-spectral radius less than $2$ are all trees, which include precisely $P_n$, $Z_n$ with $n\ge 4$, and $Z'_n$ for $n=6,7,8$. These trees with $0$-spectral radius $2$ include precisely $W_n$ with $n\ge 5$, $Z'_9$, and two additional trees with $7$ and $8$ vertices, respectively. The additional tree with $7$ vertices is a tree consisting of three pendant paths of length $2$ at a common vertex, and the additional tree with $8$ vertices is a tree consisting of two pendant paths of length $3$ and one pendant path of length $1$ at a common vertex.

Let $K_n$ be the complete graph on $n$ vertices. Let $C_n$ be the ordinary cycle  (graph) on $n\ge 3$ vertices. If the vertices of $C_n$ are labelled consecutively as $v_1,\dots, v_n, v_1$, then we write $C_n=v_1\dots v_nv_1$.

For integer $g$ with $3\leq g \leq n-1$, let $U_{n,g}$ be the unicyclic graph on $n$ vertices with girth $g$, obtained by adding an edge between one terminal vertex of the path $P_{n-g}$ and one vertex of the cycle $C_g$.

\section{Effect of local grafting operations on $\alpha$-spectral radius}

In this section, we propose some   local operations (grafting operations) that increase or decrease the $\alpha$-spectral radius of a hypergraph that is not necessarily uniform.

\begin{theorem} \label{Lem23}
Let $G$ be a connected hypergraph with $u,v\in V(G)$ and $e_{1}, \dots, e_{r}\in E(G)$ such that $u\notin e_{i}$, $v\in e_{i}$, and $(e_{i}\setminus \{v\}) \cup \{u\} \notin E(G)$ for $i=1, \dots, r$. Let $G'$ be the hypergraph obtained from $G$ by moving edges $e_{1}, \dots, e_{r}$ from $v$ to $u$. Let $x=x(G)$. If $x_{u}\geq x_{v}$, then $\rho_{\alpha}(G')> \rho_{\alpha}(G)$.
\end{theorem}

\begin{proof}  For $w\in \cup_{i=1}^r(e_{i}\setminus \{v\})$, let $s_{wv}$ be the number of edges $e_1,\dots,e_r$ containing both $w$ and $v$. For $\{w,z\}\subseteq V(G)$, it is easily seen that
\begin{align*}
&\quad a_{wz}(G')-a_{wz}(G)\\
&=\begin{cases}
s_{wv}& \mbox{if } w\in\cup_{i=1}^r(e_i\setminus \{v\}), z=u, \mbox{ or } w=u, z\in\cup_{i=1}^r(e_i\setminus \{v\})\\
-s_{wv}& \mbox{if } w\in\cup_{i=1}^r(e_i\setminus \{v\}), z=v, \mbox{ or } w=v, z\in\cup_{i=1}^r(e_i\setminus \{v\})\\
0 &\mbox{otherwise}.
\end{cases}
\end{align*}
Thus
\begin{align*}
\rho_{\alpha}(G')-\rho_{\alpha}(G)&\ge x^{\top}(A_{\alpha}(G')-A_{\alpha}(G))x\\
&= \sum_{\{w,z\}\subseteq V(G)}(a_{wz}(G')-a_{wz}(G))\left(\alpha(x_w^2+x_z^2)+2(1-\alpha)x_w x_z\right)\\
&= \sum_{w\in\cup_{i=1}^r(e_i\setminus \{v\})}s_{wv}\left(\alpha(x_w^2+x_u^2)+2(1-\alpha)x_w x_u\right)\\
&\quad -\sum_{w\in\cup_{i=1}^r(e_i\setminus \{v\})}s_{wv}\left(\alpha(x_w^2+x_v^2)+2(1-\alpha)x_w x_v\right)\\
&= \sum_{w\in\cup_{i=1}^r(e_i\setminus \{v\})}s_{wv}\left(\alpha (x_u^2-x_v^2)+2(1-\alpha)x_w(x_u-x_v)\right)\\
&=\sum_{w\in\cup_{i=1}^r(e_i\setminus \{v\})}s_{wv}(x_u-x_v)(\alpha(x_u+x_v)+2(1-\alpha)x_w).
\end{align*}
As $x_u\ge x_v$, we have $\rho_{\alpha}(G')-\rho_{\alpha}(G)\ge 0$, i.e.,
 $\rho_{\alpha}(G')\ge \rho_{\alpha}(G)$. Suppose that $\rho_{\alpha}(G')=\rho_{\alpha}(G)$. Then $\rho_{\alpha}(G')=x^{\top}A_{\alpha}(G')x$, and thus $x(G')=x$.
 Let $N_{1,2}=N_{G}(u)\cap \cup_{i=1}^r(e_i\setminus \{v\})$, $N_1=N_{G}(u)\setminus N_{1,2}$ and $N_2=\cup_{i=1}^r(e_i\setminus \{v\})\setminus N_{1,2}$.
 Note that  $N_{G'}(u)=N_1\cup N_{1,2}\cup N_2$ and $a_{uw}(G')=a_{uw}(G)$ for $w\in N_1$.
If  $N_{1,2}\ne \emptyset$, then $\sum_{w\in N_{1,2}}(a_{uw}(G')-a_{uw}(G))(\alpha x_u+(1-\alpha)x_w)>0$. If $N_{1,2}= \emptyset$, then $\sum_{w\in N_2} a_{uw}(G')(\alpha x_u+(1-\alpha)x_w)>0$.
From the eigenequations of $G$ and $G'$ at $u$, we have
\begin{align*}
0&=\rho_{\alpha}(G')x_u-\rho_{\alpha}(G)x_u\\
&= \sum_{w\in N_{G'}(u)}a_{uw}(G')(\alpha x_u+(1-\alpha)x_w)-\sum_{w\in N_{G}(u)}a_{uw}(G)(\alpha x_u+(1-\alpha)x_w)\\
&= \sum_{w\in N_{1,2}}(a_{uw}(G')-a_{uw}(G))(\alpha x_u+(1-\alpha)x_w)\\
&\quad +\sum_{w\in N_2} a_{uw}(G')(\alpha x_u+(1-\alpha)x_w)\\
&>0,
\end{align*}
a contradiction.
It follows that  $\rho_{\alpha}(G')>\rho_{\alpha}(G)$.
\end{proof}

We note that Theorem \ref{Lem23} holds for both uniform hypergraphs and hypergraphs that are not necessarily uniform.
The cases $\alpha=0, \frac{1}{2}$  have been given in \cite{LZ}.

The following theorem generalizes a result in \cite{HoSm} for $\alpha=0$, see also \cite{CRS}.

\begin{theorem} \label{Lem26}
Let $\{u,v\}$ be an edge of a connected graph $G$, and $G_{uv}$ the graph obtained from $G$ by deleting edge $\{u,v\}$ and adding  edges $\{u,w\}$ and $\{w,v\}$, where $w$ is  a new vertex not in $G$. If $\{u,v\}$ lies in some internal path of $G$, then $\rho_{\alpha}(G_{uv})\le \rho_{\alpha}(G)$ with equality if and only if $\alpha=0$ and $G\cong W_n$ with $n\ge 6$, where $n=|V(G)|$.
\end{theorem}

\begin{proof}
A label of the vertices of $G_{uv}$ is obtained from a label of the vertices of $G$ by adding the new vertex $w$.
We consider cases corresponding to the two types of the internal path.

\noindent{\bf Case 1.} $uv$ lies on an internal path $v_0 v_1 \dots v_{k+1}$ of type $(i)$,  where $v_{k+1}=v_{0}$ and $d_G(v_0)\geq 3$.

In this case, $C_{k+1}$ is a proper  induced subgraph of $G$, and thus by Lemma \ref{Lem22}, $\rho_{\alpha}(G)>\rho_{\alpha}(C_{k+1})=2$.

Relabel vertices $v_0, \dots, v_k$  as $0, 1, \dots, k$.
Let $x_0, x_1, \dots, x_k$ be the corresponding entries of the $\alpha$-Perron vector $x$ of $G$ at vertices $0, 1, \dots, k$. By Lemma \ref{Lem21}, $x_j=x_{k+1-j}$ for $j=1,\dots,k$.

\noindent {\bf Case 1.1.} $k$ is even. We may take $u=\frac{k}{2}$ and $v=\frac{k}{2}+1$. Then $x_u=x_v$.  Let $y$ be the vector  with $y_z=x_z$ for $z\in V(G)$  and $y_w=x_u$.
By the choice of $y$, $A_{\alpha}(G_{uv})y$ and $\rho_{\alpha}(G)y$ differ only in the entries corresponding to $w$, and
\begin{align*}
\left(A_{\alpha}(G_{uv})y\right)_w&=2\alpha y_w+(1-\alpha)(y_u+y_v)\\
&=2\alpha x_u+2(1-\alpha)x_u\\
&=2x_u\\
&<\rho_{\alpha}(G)x_u\\
&=\rho_{\alpha}(G)y_w.
\end{align*}
By Lemma \ref{comp}, $\rho_{\alpha}(G_{uv})< \rho_{\alpha}(G)$.

\noindent {\bf Case 1.2.} $k$ is odd. We may take $u=\frac{k-1}{2}$ and $v=\frac{k+1}{2}$.  Let $y$ be the vector  with $y_z=x_z$ for $z\in V(G)$ and $y_w=x_v$. Since $2\alpha x_v+2(1-\alpha)x_u=A_{\alpha}(G)x_v=\rho_{\alpha}(G)x_v>2x_v$, we have $x_u>x_v$. By the choice of $y$,  $A_{\alpha}(G_{uv})y$ and $\rho_{\alpha}(G)y$ differ only in the entries corresponding to $w$ and $v$, and we have
\begin{align*}
\left(A_{\alpha}(G_{uv})y\right)_w&=2\alpha y_w+(1-\alpha)(y_u+y_v)\\
&=2\alpha x_v+(1-\alpha)(x_u+x_v)\\
&<2\alpha x_v+(1-\alpha)(x_u+x_u)\\
&= \left(A_{\alpha}(G)x\right)_v\\
&=\rho_{\alpha}(G)x_v\\
&=\rho_{\alpha}(G)y_w,
\end{align*}
and
\begin{align*}
\left(A_{\alpha}(G_{uv})y\right)_v&=2\alpha y_v+(1-\alpha)(y_w+y_u)\\
&=2\alpha x_v+(1-\alpha)(x_v+x_u)\\
&<2\alpha x_v+(1-\alpha)(x_u+x_u)\\
&=\rho_{\alpha}(G)x_v\\
&=\rho_{\alpha}(G)y_v.
\end{align*}
By Lemma \ref{comp}, $\rho_{\alpha}(G_{uv})< \rho_{\alpha}(G)$.

\noindent{\bf Case 2.} $uv$ lies on an internal path $v_0v_1\dots v_{k+1}$ of type $(ii)$, where $d_G(v_0), d_G(v_{k+1})\geq 3$.

Relabel $v_0, v_1,\dots, v_{k+1}$ as $0, 1, \dots, k+1$.
Let $x_0, x_1, \dots, x_{k+1}$ be the corresponding entries of the $\alpha$-Perron vector $x$ of $G$ at $0,\dots, k+1$. Assume that $x_{0}\le x_{k+1}$. Let $t$ be the smallest index such that $x_t=\min\{x_i: 0\le i\le k+1\}$. Then $t\le k$. Assume that $u=t$ and  $v=t+1$.

Suppose first that $t>0$. Let $y$ be a vector with  $y_z=x_z$ for $z\in V(G)$ and $y_w=x_u=x_t$. Then $A_{\alpha}(G_{uv})y$ and $\rho_{\alpha}(G)y$ differ only in the entries corresponding to $w$ and $u$, and we have
\begin{align*}
\left(A_{\alpha}(G_{uv})y\right)_w&=2\alpha y_w+(1-\alpha)(y_u+y_v)\\
&=2\alpha x_u+(1-\alpha)(x_u+x_v)\\
&<2\alpha x_u+(1-\alpha)(x_{t-1}+x_v)\\
&= \left(A_{\alpha}(G)x\right)_u\\
&=\rho_{\alpha}(G)x_u\\
&=\rho_{\alpha}(G)y_w,
\end{align*}
and
\begin{align*}
\left(A_{\alpha}(G_{uv})y\right)_u&=2\alpha y_u+(1-\alpha)(y_{t-1}+y_w)\\
&=2\alpha x_u+(1-\alpha)(x_{t-1}+x_u)\\
&\le 2\alpha x_u+(1-\alpha)(x_{t-1}+x_v)\\
&=\rho_{\alpha}(G)x_u\\
&=\rho_{\alpha}(G)y_u.
\end{align*}
By Lemma \ref{comp}, $\rho_{\alpha}(G_{uv})< \rho_{\alpha}(G)$.

Suppose in the following that $t=0$. Let $S$ be the set of neighbors of $0$ other than $1$ and $s=\sum_{j\in S}x_j$. Let $d_0=d_G(0)$.

\noindent {\bf Case 2.1.}  $(1+\alpha)x_0\le \alpha d_{0}x_0+(1-\alpha)s$.

Note that $d_G(0), d_G(k+1)\ge 3$. From the result in \cite{Smi} mentioned above  and by Lemma \ref{Lem22}, we have  $\rho_0(G)\ge 2$ with equality if and only if $G\cong W_n$ with $n\ge 6$.
As $\rho_{\alpha}(G)$ is a strictly increasing function for $\alpha\in [0,1)$, we have $\rho_{\alpha}(G)\ge 2$ with equality  if and only if $\alpha=0$ and $G\cong W_n$.

Let $y$ be the vector with $y_z=x_z$ for $z\in V(G)$ and $y_w=x_u=x_0$. Then $A_{\alpha}(G_{uv})y$ and $\rho_{\alpha}(G)y$ differ only in the entries corresponding to $w$ and $u$, and we have
\begin{align*}
\left(A_{\alpha}(G_{uv})y\right)_w&=2\alpha y_w+(1-\alpha)(y_u+y_v)\\
&=2\alpha x_0+(1-\alpha)(x_0+x_1)\\
&\le \alpha d_{0}x_0+(1-\alpha)(s+x_1)\\
&=\rho_{\alpha}(G)x_0\\
&=\rho_{\alpha}(G)y_w,
\end{align*}
and
\begin{align*}
\left(A_{\alpha}(G_{uv})y\right)_u&=\alpha d_{0} y_u+(1-\alpha)(s+y_w)\\
&=\alpha d_{0} x_0+(1-\alpha)(s+x_0)\\
&\le \alpha d_{0} x_0+(1-\alpha)(s+x_1)\\
&=\rho_{\alpha}(G)x_0\\
&=\rho_{\alpha}(G)y_u.
\end{align*}
If one of these inequalities  is  strict, then  $\rho_{\alpha}(G_{uv})< \rho_{\alpha}(G)$ by Lemma \ref{comp}.
If these inequalities are equalities, then $(1+\alpha)x_0=\alpha d_{0}x_0+(1-\alpha)s$ and $x_0=x_1$, implying  $\rho_{\alpha}(G)x_0=\left(A_{\alpha}(G)x\right)_0=\alpha d_{0}x_0+(1-\alpha)(s+x_1)=(1+\alpha)x_0+(1-\alpha)x_1=2x_0$, i.e.,  $\rho_{\alpha}(G)=2$,
i.e., $\alpha=0$ and $G\cong W_n$.

\noindent {\bf Case 2.2.}   $(1+\alpha)x_0> \alpha d_{0}x_0+(1-\alpha)s$.

In this case,  $(1+\alpha-\alpha d_0)x_0>(1-\alpha)s>0$, and then  $x_0>s$, implying $s+x_0<2x_0\le (d_0-1)x_0$, i.e.,
$(1+\alpha)s+(1+\alpha-\alpha d_0)x_0<\alpha s+(1-\alpha)(d_0-1)x_0$.
For $j\in S$, we have $\rho_{\alpha}(G)x_j\ge \alpha x_{j}+(1-\alpha)x_0$, implying
\[
\rho_{\alpha}(G)s\ge \alpha s+(1-\alpha)(d_0-1)x_0>(1+\alpha)s+(1+\alpha-\alpha d_0)x_0.
\]
Thus
\[
\alpha d_{0} \frac{(1-\alpha)s}{1+\alpha-\alpha d_0}+(1-\alpha)(s+x_0)\\
<\rho_{\alpha}(G)\frac{(1-\alpha)s}{1+\alpha-\alpha d_0}.
\]

Let $y$ be the vector with $y_0=\frac{(1-\alpha)s}{1+\alpha-\alpha d_0}$,  $y_z=x_z$ for $z\in V(G)\setminus \{0\}$, and $y_w=x_u=x_0$. Then $A_{\alpha}(G_{uv})y$ and $\rho_{\alpha}(G)y$ differ only in the entries corresponding to $j\in S$, $w$ and $u$, and  for $j\in S$ we have
\begin{align*}
\left(A_{\alpha}(G_{uv})y\right)_j&=\alpha d_G(j) y_j+(1-\alpha)\sum_{i\in N_G(j)}y_i\\
&=\alpha d_G(j) x_j+(1-\alpha)\left(\sum_{i\in N_G(j)\setminus\{0\}}x_i+\frac{(1-\alpha)s}{1+\alpha-\alpha d_0}\right)\\
&<\alpha d_G(j)x_j+(1-\alpha)\left(\sum_{i\in N_G(j)\setminus\{0\}}x_i+x_0\right)\\
&=\rho_{\alpha}(G)x_j\\
&=\rho_{\alpha}(G)y_j,
\end{align*}
\begin{align*}
\left(A_{\alpha}(G_{uv})y\right)_w&=2\alpha y_w+(1-\alpha)(y_0+y_1)\\
&=2\alpha x_0+(1-\alpha)\left(\frac{(1-\alpha)s}{1+\alpha-\alpha d_0}+x_1\right)\\
&\le\alpha d_{0} x_0+(1-\alpha)(s+x_1)\\
&=\rho_{\alpha}(G)x_0\\
&=\rho_{\alpha}(G)y_w,
\end{align*}
and
\begin{align*}
\left(A_{\alpha}(G_{uv})y\right)_u&=\alpha d_{0} y_u+(1-\alpha)(s+y_w)\\
&=\alpha d_{0} \frac{(1-\alpha)s}{1+\alpha-\alpha d_0}+(1-\alpha)(s+x_0)\\
&<\rho_{\alpha}(G)\frac{(1-\alpha)s}{1+\alpha-\alpha d_0}\\
&=\rho_{\alpha}(G)y_u.
\end{align*}
Now by Lemma \ref{comp}, $\rho_{\alpha}(G_{uv})< \rho_{\alpha}(G)$.

By combining the above cases, we have either $\rho_{\alpha}(G_{uv})< \rho_{\alpha}(G)$ or $\alpha=0$ and $G\cong W_n$ with $n\ge 6$ for which we have $\rho_{\alpha}(G_{uv})=\rho_{\alpha}(G)=2$.
\end{proof}

\begin{theorem} \label{Thm31}
For $k-2\geq r \geq 2$, let $G$ be a hypergraph with two pendant edges $e_{1}=\{w_{1}, \dots, w_{k-1},u\}$ and $e_{2}=\{v_{1}, \dots, v_{r-1},u\}$ at $u$. Let $G'$ be the hypergraph obtained from $G$ by moving vertex $w_{1}$ from $e_{1}$ to $e_{2}$. Then $\rho_{\alpha}(G')< \rho_{\alpha}(G)$.
\end{theorem}

\begin{proof}  For $\{w,z\}\subseteq V(G)$,
it is easily seen that
\begin{align*}
&\quad a_{wz}(G')-a_{wz}(G)\\
&=\begin{cases}
1 & \mbox{if }  w=w_{1}, z\in \{v_{1}, \dots, v_{r-1}\}, \mbox{ or }w\in \{v_{1}, \dots, v_{r-1}\},  z=w_{1},\\
-1 & \mbox{if } w=w_{1}, z\in \{w_{2}, \dots, w_{k-1}\}, \mbox{ or }w\in \{w_{2}, \dots, w_{k-1}\},  z=w_{1},\\
0 & \mbox{otherwise}.
\end{cases}
\end{align*}
Let $x=x(G')$. By Lemma \ref{Lem21}, $x_{w_1}=x_{v_1}= \dots = x_{v_{r-1}}$ and $x_{w_2}= \dots = x_{w_{k-1}}$. Since  $\{w_2,\dots,w_{k-1},u\}\in E(G')$,  $A_{\alpha}(G'[\{w_2,\dots,w_{k-1},u\}])=A_{\alpha}(K_{k-1})$ is a principal submatrix of $A_{\alpha}(G')$. By Lemma \ref{Lem22}, we have $\rho_{\alpha}(G')> \rho_{\alpha}(K_{k-1})=k-2\geq r$. From the eigenequation of $G'$ at $w_1$ and $w_2$, we have
\[
\rho_{\alpha}(G')x_{w_1}=\alpha rx_{w_1}+(1-\alpha)((r-1)x_{w_1}+x_{u}),
\]
\[
\rho_{\alpha}(G')x_{w_2}=\alpha (k-2)x_{w_2}+(1-\alpha)((k-3)x_{w_2}+x_{u}),
\]
and thus  $(\rho_{\alpha}(G')-r-\alpha+1)x_{w_1}=(1-\alpha)x_{u}=(\rho_{\alpha}(G')-k-\alpha+3)x_{w_2}$, implying
$x_{w_1}\le x_{w_2}$. Therefore
\begin{align*}
\rho_{\alpha}(G')-\rho_{\alpha}(G)&\leq x^\top (A_{\alpha}(G')-A_{\alpha}(G)) x  \\
&= \sum_{\{w,z\}\subseteq V(G)} (a_{wz}(G')-a_{wz}(G))(\alpha(x_w^2+x_z^2)+2(1-\alpha)x_w x_z)  \\
&= \sum_{z\in \{v_{1}, \dots, v_{r-1}\}}(\alpha(x_{w_1}^2+x_z^2)+2(1-\alpha)x_{w_1} x_z)\\
&\quad -\sum_{z\in \{w_{2}, \dots, w_{k-1}\}}(\alpha(x_{w_{1}}^2+x_z^2)+2(1-\alpha)x_{w_1} x_z)\\
&= \alpha(2r-k)x_{w_1}^2-\alpha(k-2)x_{w_2}^2\\
&\quad +2(1-\alpha)x_{w_1}\left(\sum_{z\in \{v_{1}, \dots, v_{r-1}\}} x_z-\sum_{z\in \{w_{2}, \dots, w_{k-1}\}} x_z \right) \\
&=  \alpha(2r-k)x_{w_1}^2-\alpha(k-2)x_{w_2}^2+2(1-\alpha)x_{w_1}\left((r-1)x_{w_1}-(k-2)x_{w_2}\right) \\
&=  (\alpha(2r-k)+2(1-\alpha)(r-1))x_{w_1}^2-\alpha(k-2)x_{w_2}^2-2(1-\alpha)(k-2)x_{w_1}x_{w_2}\\
&\le  (\alpha(2r-k)+2(1-\alpha)(r-1)-\alpha(k-2)-2(1-\alpha)(k-2))x_{w_1}^2  \\
&=2(r-k+1)x_{w_1}^2\\
&< 0,
\end{align*}
implying   $\rho_{\alpha}(G')< \rho_{\alpha}(G)$.
\end{proof}

\begin{theorem} \label{Thm33}
Let $G, G_0, G_1, G_2$ be connected hypergraphs with $G_0$ being a cycle of length two,  where  $E(G)=E(G_0)\cup E(G_1)\cup E(G_2)$, $E(G_0)=\{e_1, e_2\}$ with $e_1\cap e_2=\{u,v\}$, $V(G_1)\cap V(G_0)=\{u\}$, $V(G_2) \cap V(G_0)=\{v\}$, and $V(G_1)\cap V(G_2)=\emptyset$.
Let $|e_i|=n_i$ for $i=1,2$. If $n_1-2\geq n_2\geq 2$, let $u_1\in e_1\backslash \{u,v\}$ and $G'$ be the hypergraph obtained from $G$ by moving vertex $u_1$ from $e_1$ to $e_2$. Then $\rho_{\alpha}(G')<\rho_{\alpha}(G)$.
\end{theorem}

\begin{proof}
Let $e_1=\{u, u_1, u_2, \dots, u_{n_1-2},v\}$, $e_2=\{u, v_1, \dots, v_{n_2-2},v\}$. For $\{w, z\}\subseteq V(G)$,
it is easily seen that
\begin{align*}
&\quad a_{wz}(G')-a_{wz}(G)\\
&=\begin{cases}
1 & \mbox{if } w=u_{1}, z\in \{v_{1}, \dots, v_{n_2-2}\}, \mbox{ or }w\in \{v_{1}, \dots, v_{n_2-2}\}, z=u_{1},\\
-1 & \mbox{if } w=u_{1}, z\in \{u_{2}, \dots, u_{n_1-2}\}, \mbox{ or }w\in \{u_{2}, \dots, u_{n_1-2}\}, z=u_{1},\\
0 & \mbox{otherwise}.
\end{cases}
\end{align*}
Let $x=x(G')$. By Lemma \ref{Lem21}, $x_{u_1}=x_{v_1}= \dots = x_{v_{n_2-2}}$ and $x_{u_2}= \dots = x_{u_{n_1-2}}$. By Lemma \ref{Lem22}, we also have $\rho_{\alpha}(G')> \rho_{\alpha}(K_{n_1-1})=n_1-2\geq n_2$. From the eigenequation of $G'$ at $u_1$ and $u_2$, we have
\begin{align*}
\rho_{\alpha}(G')x_{u_1}&=\alpha n_2x_{u_1}+(1-\alpha)\left(x_{u}+x_{v}+(n_2-2)x_{u_1}\right),\\
\rho_{\alpha}(G')x_{u_2}&=\alpha (n_1-2)x_{u_2}+(1-\alpha)\left(x_{u}+x_{v}+(n_1-4)x_{u_2}\right),
\end{align*}
and thus  $(\rho_{\alpha}(G')-n_2+2-2\alpha)x_{u_1}=(1-\alpha)(x_{u}+x_{v})=(\rho_{\alpha}(G')-n_1+4-2\alpha)x_{u_2}$, implying  $x_{u_1}\leq x_{u_2}$.
Therefore
\begin{align*}
\rho_{\alpha}(G')-\rho_{\alpha}(G)&\leq  x^\top (A_{\alpha}(G')-A_{\alpha}(G)) x \\
&= \sum_{\{w,z\}\subseteq V(G)} (a_{wz}(G')-a_{wz}(G))(\alpha(x_w^2+x_z^2)+2(1-\alpha)x_w x_z)\\
&= \sum_{z\in \{v_{1}, \dots, v_{n_2-2}\}}(\alpha(x_{u_1}^2+x_z^2)+2(1-\alpha)x_{u_1} x_z)\\
&\quad -\sum_{z\in \{u_{2}, \dots, u_{n_1-2}\}}(\alpha(x_{u_{1}}^2+x_z^2)+2(1-\alpha)x_{u_1} x_z)\\
&= \alpha(2n_{2}-n_{1}-1)x_{u_1}^2-\alpha(n_1-3)x_{u_2}^2\\
&\quad +2(1-\alpha)x_{u_1}\left(\sum_{z\in \{v_{1}, \dots, v_{n_2-2}\}} x_z-\sum_{z\in \{u_{2}, \dots, u_{n_1-2}\}}x_z \right)\\
&= \alpha(2n_{2}-n_{1}-1)x_{u_1}^2-\alpha(n_1-3)x_{u_2}^2\\
&\quad +2(1-\alpha)x_{u_1}\left((n_2-2)x_{u_1}-(n_1-3)x_{u_2}\right)\\
&= (\alpha(2n_{2}-n_{1}-1)+2(1-\alpha)(n_2-2))x_{u_1}^2-\alpha(n_1-3)x_{u_2}^2\\
&\quad -2(1-\alpha)(n_1-3)x_{u_1}x_{u_2}\\
&\le  (\alpha(2n_{2}-n_{1}-1)+2(1-\alpha)(n_2-2)-\alpha(n_1-3)-2(1-\alpha)(n_1-3))x_{u_1}^2\\
&=2(n_2-n_1+1)x_{u_1}^2 \\
&< 0,
\end{align*}
implying $\rho_{\alpha}(G')< \rho_{\alpha}(G)$.
\end{proof}

\begin{theorem} \label{Thm34}
For $k\geq 3$, let $e=\{w_1, \dots,w_k\}$ be an edge of a connected hypergraph $G$, and $e\setminus \{w_2\}$, $\{w_1,w_2\}$ $\notin E(G)$. Let $G'$ be the hypergraph obtained from $G$ by removing $w_2$ from $e$ and attaching an edge $\{w_1,w_2\}$ to $w_1$. Then $\rho_{\alpha}(G')< \rho_{\alpha}(G)$.
\end{theorem}

\begin{proof}
Note that $a_{wz}(G')-a_{wz}(G)\leq 0$ for $\{w,z\}\subseteq V(G)$, and $a_{w_2w_k}(G')- a_{w_2w_k}(G)=-1$.  
Then   $A(G')<A(G)$, and thus  $A_{\alpha}(G')<A_{\alpha}(G)$. By Lemma \ref{Lem22}, $\rho_{\alpha}(G')< \rho_{\alpha}(G)$.
\end{proof}

\section{Extremal $\alpha$-spectral radius of hypertrees}


In this section, we determine the unique hypertrees with maximum and/or minimum  $\alpha$-spectral radius among some classes of hypertrees. 

\subsection{The $\alpha$-spectral radius of uniform hypertrees}

Let $D_{n,k,c}$ be the $k$-uniform hypertree obtained from $S_{k,k}$ by attaching $c$ pendant edges at one vertex and $\frac{n-1}{k-1}-1-c$ pendant edges at another vertex, where $\frac{n-1}{k-1}\geq 3 $ and $1\leq c\leq \left\lfloor\frac{n-k}{2(k-1)}\right\rfloor$.

\begin{theorem}\label{NSmax}
Let $T$ be a $k$-uniform hypertree of order $n$, where $2\le k\le n$.
Then $\rho_{\alpha}(T)\leq \frac{n\alpha+k-2+\sqrt{(n\alpha+k-2)^2-4(n-1)(k\alpha-1)}}{2}$ with equality if and only if $T\cong S_{n,k}$. Moreover, if $T\ncong S_{n,k}$ with $\frac{n-1}{k-1}\ge4$, then
 $\rho_{\alpha}(T)\le \rho (D_{n,k,1})$ with equality if and only if $T\cong D_{n,k,1}$.
%
\end{theorem}

\begin{proof}
Let $T$ be a $k$-uniform hypertree of order $n$ with maximum $\alpha$-spectral radius. Let $d$ be the diameter of $T$ and $u_0 e_1 u_1\ldots u_{d-1}e_du_d$ a diametrical path   in $T$. Suppose that $d\geq 3$. Let $x=x(T)$. Assume that  $x_{u_1}\geq x_{u_{d-1}}$. Let $T'$ be the $k$-uniform hypertree obtained from $T$ by moving all edges in $E_T(u_{d-1})\setminus\{e_{d-1}\}$ from $u_{d-1}$ to $u_1$. By Theorem \ref{Lem23}, $\rho_{\alpha}(T')>\rho_{\alpha}(T)$, a contradiction.
Thus $d=2$ and $T\cong S_{n,k}$.

Let $w$ be the center of $S_{n,k}$, and $c$ the entry of $x(S_{n,k})$ at $w$. By Lemma \ref{Lem21}, for $u,v\in V(S_{n,k})\setminus \{w\}$, $x_u=x_v$, which we denote by $p$.
From the eigenequation of $S_{n,k}$ at $w$ and a vertex of degree $1$ in $S_{n,k}$,  we have
\begin{align*}
\left(\rho_{\alpha}(S_{n,k})-(n-1)\alpha\right)c-(1-\alpha)(n-1)p&=0,\\
-(1-\alpha)c+\left(\rho_{\alpha}(S_{n,k})-\alpha-k+2\right)p&=0.
\end{align*}
As $c,p\ne 0$, we have
\[
\det \begin{pmatrix}
\rho_{\alpha}(S_{n,k})-(n-1)\alpha  &-(1-\alpha)(n-1)\\
-(1-\alpha) & \rho_{\alpha}(S_{n,k})-\alpha-k+2
\end{pmatrix} =0,
\]
i.e., $f(\rho_{\alpha}(S_{n,k}))=0$, where
$f(t)=t^2-(n\alpha+k-2)t+(n-1)(k\alpha-1)$.
Then $\rho_{\alpha}(S_{n,k})$ is the largest root of $f(t)=0$, i.e.,  $\rho_{\alpha}(S_{n,k})=\frac{n\alpha+k-2+\sqrt{(n\alpha+k-2)^2-4(n-1)(k\alpha-1)}}{2}$.
This proves the first part.

Next, we prove the second part.  Suppose that $T\ncong S_{n,k}$ with $\frac{n-1}{k-1}\ge 4$.
Let $T$ be a hypertree nonisomorphic to $S_{n,k}$ with maximum $\alpha$-spectral radius among $k$-uniform hypertrees of order $n$.
Let $d$ be the diameter of $T$.  As $T\ncong S_{n,k}$, we have $d\ge 3$.
By similar argument as above, we have
$d =3$. Obviously, $T$ has a unique edge, say $e=\{w_1,\ldots,w_k\}$, which is not a pendant edge, and  $T$ is  obtainable   by
attaching some pendant edges at vertices in $e$.
If there are three vertices $w_i, w_j, w_{\ell}$ with degree at least $2$ in $T$, then we may move all pendant edges at $w_j$ to $w_i$ or vice versa to obtain a $k$-uniform hypertree $T''$ of order $n$ with diameter $3$, and by Theorem \ref{Lem23}, $\rho_{\alpha}(T'')>\rho_{\alpha}(T)$, a contradiction. Thus, among vertices $w_1, \dots, w_k$, there are exactly two vertices, say $w_1$ and $w_2$, with degree at least $2$ in $T$.
It follows that $T\cong D_{n,k,c}$ for some  $c$ with $1\leq c\leq \left\lfloor\frac{n-k}{2(k-1)}\right\rfloor$.
Now by Theorem \ref{Lem23}, we have
 $=1$, and thus $T\cong D_{n,k,1}$.
\end{proof}

Let $G$ be a $k$-uniform hypergraph with $e_{1},e_{2}\in E(G)$ and $u,v\in V(G)$ such that $u\in e_{1}\setminus e_{2}$ and $v\in e_{2}\setminus e_{1}$. Let $e'_1=(e_{1}\setminus\{u\})\cup \{v\}$ and $e'_2=(e_{2}\setminus\{v\})\cup \{u\}$. Suppose that $e'_1 ,e'_2\notin E(G)$. Let $G'$ be the hypergraph with $V(G')=V(G)$ and $E(G')=(E(G)\setminus \{e_{1}, e_{2}\})\cup \{e'_{1}, e'_{2}\}$. Then we say that $G'$ is obtained from $G$ by exchanging vertices $u$ in $e_1$ and $v$ in $e_{2}$.

For $k\geq 3$, $\frac{n-1}{k-1}\geq 4$ and a hyperstar $S_{k,k}$ with edge $e=\{w_1,\ldots,w_k\}$,
let $H_{n,k}$ be the $k$-uniform hypertree obtained from $e$ by attaching $\frac{n-1}{k-1}-3$ pendant edges at $w_1$ and attaching
a pendant edge at $w_2$ and $w_3$, respectively. Let $H_{4,2}=P_{4,2}$.


\begin{theorem} \label{FG22} Let $T$ be a $k$-uniform hypertree of order $n$ and $T\ncong S_{n,k}, D_{n,k,1}$ with maximum $\alpha$-spectral radius, where $2\le k\le n$. Then
\[
T\cong \begin{cases}
P_{5,2} & \mbox{if $(n,k)=(5,2)$},\\
H_{n,k} & \mbox{if $\frac{n-1}{k-1}=4$ with $k\ge 3$},\\
D_{n,k,2} & \mbox{if  $\frac{n-1}{k-1}\ge 5$}.
\end{cases}
\]
\end{theorem}

\begin{proof} Suppose first that $\frac{n-1}{k-1}=4$.
As $T\ncong S_{n,k}, D_{n,k,1}$, we have $T\cong P_{5,2}$ if $k=2$, and  $T\cong P_{n,k}, H_{n,k}$ if $k\ge 3$. Thus we need only to show that $\rho_{\alpha}(P_{4k-3,k}) <\rho_{\alpha}(H_{4k-3,k})$ for $k\ge 3$.

Let $P_{4k-3,k}=u_0 e_1 u_1 e_2 u_2 e_3 u_3 e_4 u_4$ and $x=x(P_{4k-3,k})$. By Lemma \ref{Lem21}, $x_{u_1}=x_{u_3}$, $x_u=x_v$ for $u,v\in (e_1\cup e_4)\setminus \{u_1,u_3\}$, and $x_w=x_z$ for $w,z\in (e_2\cup e_3)\setminus \{u_1,u_2,u_3\}$. Let $w\in e_2\setminus\{u_1,u_2\}$.
Suppose that  $x_w\ge x_{u_3}$. Then we form a $k$-uniform hypertree $T'$ from $P_{4k-3,k}$ by moving $e_4$ from $u_{3}$ to $w$.  Obviously, $T'\cong H_{4k-3,k}$. By Theorem \ref{Lem23}, $\rho_{\alpha}(P_{4k-3,k})<\rho_{\alpha}(H_{4k-3,k})$, as desired. Suppose that  $x_w< x_{u_3}$. Then we form a $k$-uniform hypertree $T''$ from $P_{4k-3,k}$ by exchange vertices $w$ in $e_2$ and $u_3$ in $e_3$. Obviously,  $T''\cong H_{4k-3,k}$. Note that
\begin{align*}
&\quad a_{uv}(T'')-a_{uv}(P_{4k-3,k})\\
&=\begin{cases}
1 & \mbox{if } u=u_{3}, v\in e_{2}\setminus \{u_{2},w\}, \mbox{ or }u\in e_{2}\setminus \{u_{2},w\}, v=u_{3},\\
1 & \mbox{if } u=w, v\in e_{3}\setminus \{u_{2},u_{3}\}, \mbox{ or }u\in e_{3}\setminus \{u_{2},u_{3}\}, v=w,\\
-1 & \mbox{if } u=u_{3}, v\in e_{3}\setminus \{u_{2},u_{3}\}, \mbox{ or }u\in e_{3}\setminus \{u_{2},u_{3}\}, v=u_{3},\\
-1 & \mbox{if } u=w, v\in e_{2}\setminus \{u_{2},w\}, \mbox{ or }u\in e_{2}\setminus \{u_{2},w\}, v=w,\\
0 & \mbox{otherwise}.
\end{cases}
\end{align*}
Thus
\begin{align*}
&\quad \rho_{\alpha}(T'')-\rho_{\alpha}(P_{4k-3,k})\\
&\ge x^\top(A_{\alpha}(T'')-A_{\alpha}(P_{4k-3,k}))x\\
&= \sum_{\{u,v\}\subseteq V(T)} (a_{uv}(T'')-a_{uv}(P_{4k-3,k}))(\alpha(x_u^2+x_v^2)+2(1-\alpha)x_u x_v)\\
&= \sum_{v\in e_{2}\setminus \{u_{2},w\}}\left(\alpha(x_{u_3}^2+x_v^2)+2(1-\alpha)x_{u_3} x_v\right)\\
&\quad +\sum_{v\in e_{3}\setminus \{u_{2},u_{3}\}}\left(\alpha(x_{w}^2+x_v^2)+2(1-\alpha)x_{w} x_v\right)\\
&\quad -\sum_{v\in e_{3}\setminus \{u_{2},u_{3}\}}\left(\alpha(x_{u_{3}}^2+x_v^2)+2(1-\alpha)x_{u_3} x_v\right)\\
&\quad -\sum_{v\in e_{2}\setminus \{u_{2},w\}}\left(\alpha(x_{w}^2+x_v^2)+2(1-\alpha)x_{w} x_v\right)\\
&=\sum_{v\in e_{2}\setminus \{u_{2},w\}}\left(\alpha(x_{u_3}^2-x_w^2)+2(1-\alpha)x_v(x_{u_3}-x_w)\right)\\
&\quad +\sum_{v\in e_{3}\setminus \{u_{2},u_{3}\}}\left(\alpha(x_w^2-x_{u_3}^2)+2(1-\alpha)x_v(x_w-x_{u_3})\right)\\
&=2(1-\alpha)(x_{u_3}-x_w)\left(\sum_{v\in e_{2}\setminus \{u_{2},w\}}x_v-\sum_{v\in e_{3}\setminus \{u_{2},u_{3}\}}x_v\right)\\
&=2(1-\alpha)(x_{u_3}-x_w)^2\\
&> 0,
\end{align*}
implying  $\rho_{\alpha}(P_{4k-3,k})<\rho_{\alpha}(H_{4k-3,k})$, as desired. Therefore, $T\cong H_{n,k}$.

Suppose in the following that  $\frac{n-1}{k-1}\ge 5$ and $x=x(T)$. Let $d$ be the diameter of $T$. 
By similar argument as in the proof of Theorem \ref{NSmax}, we have $d=3$, or $d=4$ and $T\cong T_{n,k,s}$ with $1\le s \le \lfloor \frac{n-1}{2(k-1)}-1\rfloor$, where
 $T_{n,k,s}$ is the $k$-uniform hypertree  obtained from $S_{2k-1,k}$  by attaching $s$ pendant edges at a vertex of degree one in one edge and $\frac{n-1}{k-1}-2 -s$ pendant edges at a vertex of degree one in the other edge.
Suppose that  the latter case occurs.
Let $P=u_0 e_1 u_1 e_2 u_2 e_3 u_3 e_4 u_4$ be a diametrical path in $T$. If  $s\ge 2$, then,  by Theorem \ref{Lem23}, we may move all edges of $E_T(u_3)\setminus\{e_3, e_4\}$ from $u_3$ to $u_1$ or move all edges of $E_T(u_1)\setminus\{e_1, e_2\}$ from $u_1$ to $u_3$ to form a $k$-uniform hypertree  with larger $\alpha$-spectral radius, which is impossible.
Thus  $s=1$.
Let $T'$ be the $k$-uniform hypertree obtained from $T$ by moving $e_1$ from $u_1$ to $u_2$ or by moving $e_3$ from $u_2$ to $u_1$. Obviously, $T'\ncong S_{n,k}, D_{n,k,1}$. By Theorem \ref{Lem23}, $\rho_{\alpha}(T')>\rho_{\alpha}(T)$, a contradiction.
It follows that $d=3$, and $T$ is a $k$-uniform hypertree obtainable from a hyperstar $S_{k,k}$ with a single edge, say $e =\{w_1,\dots, w_k\}$, by attaching $a_i$ pendant edges at $w_i$, where $\sum_{i=1}^ka_i+1=\frac{n-1}{k-1}$, $a_i\ge0$ for $1\le i\le k$ and $a_p, a_q\ge1$ for some $p$ and $q$ with $1\le p <q\le k$. By relabeling the vertices in $e$, we may assume that $a_2\ge a_1\ge a_3\ge\dots\ge a_k$. Obviously, $a_1\ge1$.
Suppose that  $a_3\ge 1$. Then we may move all the  pendant edges at $w_1$ from $w_1$ to $w_3$ or move all the  pendant edges at $w_3$ from $w_3$ to $w_1$ to form a $k$-uniform hypertree $T''$.  Obviously, $T''\ncong S_{n,k}, D_{n,k,1}$. By Theorem \ref{Lem23}, $\rho_{\alpha}(T'')>\rho_{\alpha}(T)$, a contradiction. Thus $a_3=0$. Thus, for $k\ge 2$, $T\cong D_{n,k,a_1}$ with $a_1\ge 2$. By Theorem \ref{Lem23}, $T\cong D_{n,k,2}$.
%
\end{proof}


\subsection{The $\alpha$-spectral radius of hypertrees that are not necessarily uniform}


For $1\leq m\leq n-1$, let $S^m_n$ be the hyperstar on $n$ vertices with $m-1$ edges of size $2$ and one edge of size $n-m+1$. Particularly, $S_n^1$ consists of a single edge with size $n$.

\begin{theorem} \label{Thm41}
Let $T$ be a hypertree on $n$ vertices with $m$ edges, where $1\leq m\leq n-1$. Then $\rho_{\alpha}(T)\leq \rho_{\alpha}(S^m_n)$ with equality if and only if $T\cong S^m_n$.
\end{theorem}

\begin{proof}
It is trivial if $m=1$.
Suppose that $m\geq 2$. Let $T$ be a hypertree on $n$ vertices with $m$ edges having maximum $\alpha$-spectral radius.

Suppose that there is an edge $e\in E(T)$, which has two vertices, say $v_1$ and $v_2$,  of degree at least $2$. Since $T$ is a hypertree,
$(E_T(v_1)\setminus \{e\})\cap (E_T(v_2)\setminus \{e\})=\emptyset$. Let $x=x(T)$. We may assume that $x_{v_1}\geq x_{v_2}$. Let $T'$ be the hypergraph obtained from $T$ by moving all edges containing $v_2$ except $e$ from $v_2$ to $v_1$. Obviously, $T'$ is a hypertree on $n$ vertices with $m$ edges.  By Theorem \ref{Lem23}, $\rho_{\alpha}(T')> \rho_{\alpha}(T)$, a contradiction. Thus all the edges in $T$ are pendant edges at a common vertex, i.e., $T$ is a hyperstar.

Suppose that there are two edges, say  $e_1$ and $e_2$, of size at least $3$.  Assume that $|e_1|\geq |e_2|$. Let $T''$ be the hypertree obtained from $T$  by moving a vertex $w\in e_2\setminus e_1$ from $e_2$ to $e_1$. By Theorem \ref{Thm31}, $\rho_{\alpha}(T'')> \rho_{\alpha}(T)$, a contradiction. It follows that there is at most one edge of size at least $3$, and thus $T\cong S^m_n$.
\end{proof}

\begin{lemma} \label{Lem41}
If $1\leq m_2 <m_1 \leq n-1$, then $\rho_{\alpha}(S^{m_1}_n)< \rho_{\alpha}(S^{m_2}_n)$.
\end{lemma}

\begin{proof}
Let $T=S^{m_2}_n$ with center $u$ and let  $e\in E(T)$ of size $n-m_2+1$. Let $T'$ be the hypertree obtained from $T$ by removing one vertex, say $w$, in $e\setminus \{u\}$ from $e$ and attaching a pendant edge $\{u,w\}$ to $u$. Then $T' \cong S^{m_2+1}_n$. By Theorem \ref{Thm34}, $\rho_{\alpha}(S^{m_2+1}_n)< \rho_{\alpha}(S^{m_2}_n)$. By applying this process repeatedly, we finally have $\rho_{\alpha}(S^{m_1}_n)< \rho_{\alpha}(S^{m_2}_n)$.
\end{proof}

\begin{corollary} \label{Cor41}
Let $T$ be a hypertree on $n$ vertices with maximum degree $\Delta$, where $1\leq \Delta \leq n-1$. Then $\rho_{\alpha}(T)\leq \rho_{\alpha}(S^\Delta_n)$ with equality if and only if $T\cong S^\Delta_n$.
\end{corollary}

\begin{proof}
Note that maximum degree of $S^\Delta_n$ is $\Delta$ and  $|E(T)|\geq \Delta$. By Theorem \ref{Thm41} and Lemma \ref{Lem41}, we have $\rho_{\alpha}(T)\leq \rho_{\alpha} \left(S^{|E(T)|}_n\right) \leq \rho_{\alpha}(S^\Delta_n)$ with equalities if and only if $T\cong S^{|E(T)|}_n$ and $|E(T)|=\Delta$, i.e., $T\cong S^\Delta_n$.
\end{proof}

Similarly, we have

\begin{corollary} \label{Cor42}
Let $T$ be a hypertree on $n$ vertices with $p$ pendant edges, where $2\leq p \leq n-1$. Then $\rho_{\alpha}(T)\leq \rho_{\alpha}(S^p_n)$ with equality if and only if $T\cong S^p_n$.
\end{corollary}


%
%
%
%

\begin{theorem} \label{Thm43}
Let $T$ be a hypertree on $n$ vertices. Then $\rho_{\alpha}(T)\leq n-1$ with equality if and only if $T\cong S^1_n$.
Moreover, if  $T\ncong S^1_n$, then
$\rho_{\alpha}(T)\leq \rho_{\alpha}(S^2_n)$ with equality if and only if $T\cong S^2_n$.
\end{theorem}

\begin{proof}
By Theorem \ref{Thm41} and Lemma \ref{Lem41}, $\rho_{\alpha}(T)\leq \rho_{\alpha} \left(S^{|E(T)|}_n\right) \leq \rho_{\alpha}(S^1_n)$ with equalities if and only if $T\cong S^{|E(T)|}_n$ and $|E(T)|=1$, i.e., $T\cong S^1_n$. The first part follows by noting that $\rho_{\alpha}(S^1_n)=\rho_{\alpha}(K_n)=n-1$.

If  $T\ncong S^1_n$, then $|E(T)|\geq 2$, and thus by Theorem \ref{Thm41} and Lemma \ref{Lem41}, $\rho_{\alpha}(T)\leq \rho_{\alpha} \left(S^{|E(T)|}_n\right) \leq \rho_{\alpha}(S^2_n)$ with equalities if and only if $T\cong S^{|E(T)|}_n$ and $|E(T)|=2$, i.e., $T\cong S^2_n$. The second part follows.
\end{proof}

We note that the first part follows also from Lemma \ref{Lem22}.

For $n\geq 5$, let $S'^2_n$ be the hyperstar on $n$ vertices with one edge of size $3$ and one edge of size $n-2$.

\begin{theorem} \label{Thm44}
Let $T$ be a hypertree on $n\geq 5$ vertices, where $T\ncong S^1_n,S^2_n$. Then $\rho_{\alpha}(T)\leq \rho_{\alpha}(S'^2_n)$ with equality if and only if $T\cong S'^2_n$.
\end{theorem}

\begin{proof}
Let $m=|E(T)|$. Since $T\ncong S^1_n,S^2_n$, we have $m\geq 2$. If $m\geq 3$, then we have by Theorem \ref{Thm41} and Lemma \ref{Lem41} that $\rho_{\alpha}(T)\leq \rho_{\alpha}(S^{m}_n) \leq \rho_{\alpha}(S^3_n)$ with equalities if and only if $T\cong S^{m}_n$ and $m=3$, i.e., $T\cong S^3_n$. If $m=2$, then by Theorem \ref{Thm31}, $\rho_{\alpha}(T)\leq \rho_{\alpha}(S'^{2}_n)$. By Theorem \ref{Thm34}, we have $\rho_{\alpha}(S^3_n) < \rho_{\alpha}(S'^{2}_n)$.
\end{proof}

If $T$ is a hypertree on $4$ vertices, where $T\ncong S^1_4,S^2_4$, then $T\cong P_4$ or $S_4$. By Theorem \ref{Lem23}, $\rho_{\alpha}(S_4)>\rho_{\alpha}(P_4)$.



Let $P'_n$ be the hypertree on $n\geq 4$ vertices obtained from $P_{n-2}=v_1\dots v_{n-2}$ by adding an edge $\{v_{n-2},v_{n-1},v_{n}\}$ with size $3$. Then $A_{\alpha}(P'_n)$ is permutation similar to $A_{\alpha}(U_{n,3})$, which implies  $\rho_{\alpha}(P'_n)=\rho_{\alpha}(U_{n,3})$.

\begin{theorem} \label{Thm54}
Let $T$ be a hypertree on $n\geq 4$ vertices that is not $2$-uniform. Then $\rho_{\alpha}(T)\geq \rho_{\alpha}(P'_n)$ with equality if and only if $T\cong P'_n$.
\end{theorem}

\begin{proof}
Let $T$ be a hypertree on $n$ vertices that is not $2$-uniform  with minimum $\alpha$-spectral radius.

Suppose that  $T$ has two edges, say $e$ and $e'$,  of size at least $3$. Since $T$ is a hypertree, we may assume that $\{u,v,w\}\subseteq e$ with $u,v\notin e'$.
Let $T'$ be the hypertree obtained from $T$ by removing $v$ from $e$ and attaching an edge $\{u,v\}$ to $u$. Obviously, $T'$ is not $2$-uniform. By Theorem \ref{Thm34}, $\rho_{\alpha}(T)>\rho_{\alpha}(T')$, a contradiction. Thus $T$ has exactly one edge of size at least $3$. By similar argument as above, $T$ has no edge of size at least $4$. It follows that $T$ has exactly one edge of size $3$ and all other edges are of size $2$. Let
 $e^*=\{w_1,w_2,w_3\}$ be the unique edge of size $3$ of $T$. Let $U$ be the hypergraph obtained from $T$ by deleting edge $e^*$ and adding edges $\{w_1,w_2\}$, $\{w_1,w_3\}$, and $\{w_2,w_3\}$. We can easily see that $U$ is a unicyclic graph with girth $3$. Note that $A_{\alpha}(T)$ is permutation similar to $A_{\alpha}(U)$. The deletion of $e^*$ from $T$ yields trees $T_1$, $T_2$ and $T_3$, where $w_i\in V(T_i)$ for $i=1,2,3$.
By applying Lemma \ref{Lem24} to $U$, we know that $T_i$ is a path with a terminal vertex $w_i$ for $i=1,2,3$. By Lemma \ref{Lem25}, only one of $T_1$, $T_2$ and $T_3$ is nontrivial. Thus $T\cong P'_n$.
\end{proof}

Let $T$ be a hypertree on $n$ vertices. If  $T$ is not $2$-uniform, then by Theorem \ref{Thm54}, we have  $\rho_{\alpha}(T)\ge \rho_{\alpha}(P'_n)$. If  $T$ is $2$-uniform and $T\ncong P_n$, then by Lemma \ref{Lem24},  we have $\rho_{\alpha}(T)\ge\rho_{\alpha}(Z_n)>\rho_{\alpha}(P_n)$ with equality if and only if $T\cong Z_n$. By Lemma \ref{Lem22},
$\rho_{\alpha}(Z_n)<\rho_{\alpha}(U_{n,3})=\rho_{\alpha}(P'_n)$.
Thus, among hypertrees on $n$ vertices,
\begin{enumerate}
\item[(a)] $P_n$ is the unique hypertree with smallest $\alpha$-spectral radius;
\item[(b)]  for $n\ge 4$, $Z_n$ is the unique tree with second smallest $\alpha$-spectral radius.
\end{enumerate}

\section{Extremal $\alpha$-spectral radius of unicyclic hypergraphs}

In this section, we study the  $\alpha$-spectral radius of a unicyclic hypergraph.

We consider $k$-uniform unicyclic hypergraphs first. For a $k$-uniform unicyclic hypergraph $G$ with $V(G)=\{v_1, \dots, v_n\}$,  if $E(G)=\{e_1, \dots, $$e_m\}$,
where $e_i=\{v_{(i-1)(k-1)+1}, \dots,$ $v_{(i-1)(k-1)+k}\}$ for $i=1,\ldots,m$ and $v_{(m-1)(k-1)+k}=v_1$,
then we call $G$ a $k$-uniform loose cycle, denoted by $C_{n,k}$.
For integers $k\geq 2$, $g\ge 2$, $a\ge 1$,  and a $k$-uniform loose cycle $C_{g(k-1),k}$,
let $C^k_g(a)$ be the $k$-uniform hypergraph obtained from $C_{g(k-1),k}$ by
identifying the center of $S_{a(k-1)+1,k}$ and $v_1$.
%
%
For $k\geq 3$ and $\frac{n}{k-1}\geq g\geq 2$, let $F_{n,k,g}$ be a $k$-uniform unicyclic hypergraph obtained from $C_{g(k-1),k}$ with edges $e_i=\{v_{(i-1)(k-1)+1}, \dots, v_{(i-1)(k-1)+k}\}$ by attaching $\frac{n}{k-1}-g$ pendant edges at $v_2$, where $i=1,\ldots,g$ and $v_{g(k-1)+1}=v_1$.

\begin{theorem} \label{hypercycle-sign-girth-max0} Let $G$ be a $k$-uniform unicyclic hypergraph with order $n$ and girth $g$, where $2\le k\le n$.
Then $\rho_{\alpha}(G)\leq\rho_{\alpha}(C^k_g(m-g))$ with equality if and only if $G\cong C^k_g(m-g)$, where $m=\frac{n}{k-1}$.
\end{theorem}

\begin{proof}  It is trivial if $g=m$. Suppose that $g<m$.  Let $C$ be the unique cycle of $G$ and $v$ a vertex in $C$.  Suppose that $e$ is an edge $e$ outside $C$ contains $v$. If $e$ is not a pendant edge at $v$, then there is another vertex in $e$, say $w$, with degree at least $2$ in $G$, and we may move all edges containing $v$ except $e$ from $v$ to $w$ or  vice versa to obtain
a $k$-uniform unicyclic hypergraph $G'$ with order $n$ and girth $g$, for which we have by   Theorem \ref{Lem23} that $\rho_{\alpha}(G')>\rho_{\alpha}(G)$, a contradiction. Thus every edge outside $C$ is a pendant edge at some vertex in $C$. If there are at least two vertices in $C$ at which there are pendant edges, then we may choose two such vertices, and  move all pendant edges from one vertex to the other to obtain a $k$-uniform unicyclic hypergraph $G''$ with order $n$ and girth $g$, for which we have by   Theorem \ref{Lem23} that $\rho_{\alpha}(G'')>\rho_{\alpha}(G)$, also a contradiction. Thus all pendant edges are at a common vertex of $C$.
Therefore $G\cong C^k_g(m-g)$ if $k=2$, and  $G\cong C^k_g(m-g)$ or $G\cong F_{n,k,g}$ if $k\ge 3$.

Suppose that $k\ge 3$ and $G\cong F_{n,k,g}$. Let $G^*$ be the $k$-uniform hypergraph obtained from $G$ by moving all edges of $E_G(v_2)\setminus\{e_1\}$ from $v_2$ to $v_1$, and $G^{**}$ be the $k$-uniform hypergraph obtained from $G$ by moving $e_g$ from $v_1$ to $v_2$.
Obviously,  $G^*\cong G^{**}\cong G_g^k(m-g)$. 
By Theorem \ref{Lem23}, we have $\rho_{\alpha}(G^*)>\rho_{\alpha}(G)$ or $\rho_{\alpha}(G^{**})>\rho_{\alpha}(G)$, a contradiction.
Thus $G\cong C^k_g(m-g)$.
\end{proof}

\begin{theorem} \label{hypercycle-sign-max0} Let $G$ be a $k$-uniform unicylic hypergraph of order $n$, where $2\le k\le n$. Let $m=\frac{n}{k-1}$.
Then

(i) if $G$ is linear, then $m\geq 3$ and $\rho_{\alpha}(G)\leq\rho_{\alpha}(C^k_3(m-3))$ with equality if and only if $G\cong C^k_3(m-3)$;

(ii) if $k\geq 3$, then $m\geq 2$ and $\rho_{\alpha}(G)\leq\rho_{\alpha}(C^k_2(m-2))$ with equality if and only if $G\cong C^k_2(m-2)$.
\end{theorem}

\begin{proof}
Let $G$ be a $k$-uniform unicyclic hypergraph of order $n$ with maximum $\alpha$-spectral radius.

Let $g$ be the girth of $G$. Let $v_1 e_1 v_2 \ldots v_g e_g v_1$ be the unique cycle of $G$. Let $x=x(G)$.

Suppose first that $G$ is linear. Then $m\ge g\ge 3$.
Suppose that $g\geq 4$. Assume that $x_{v_1}\geq x_{v_4}$.
Let $G'$ be the $k$-uniform unicyclic hypergraph obtained from $G$ by moving all edges of $E_G(v_4)\setminus \{e_4\}$
from $v_4$ to $v_1$. Obviously, $G'$ is linear.
By Theorem \ref{Lem23}, we have $\rho_{\alpha}(G')>\rho_{\alpha}(G)$, a contradiction. Thus $g=3$. By Theorem \ref{hypercycle-sign-girth-max0}, $G\cong C^k_3(m-3)$.

Suppose that $k\geq 3$. Then $m\geq g\geq 2$. Suppose that $g\geq 3$.  We may assume that $x_{v_1}\geq x_{v_3}$. Let $G''$ be the $k$-uniform unicyclic hypergraph obtained from $G$  by moving all edges of $E_G(v_3)\setminus \{e_3\}$ from $v_3$ to $v_1$. 
By Theorem \ref{Lem23}, we have $\rho_{\alpha}(G'')>\rho_{\alpha}(G)$, a contradiction. Thus $g=2$, and by Theorem \ref{hypercycle-sign-girth-max0}, we have $G\cong C^k_2(m-2)$. 
\end{proof}

Now we move to consider unicyclic hypergraphs that are not necessarily uniform.
Let $U^1_n$ be the unicyclic hypergraph on $n\geq 3$ vertices obtained from $S^1_n$ by adding an edge of size $2$. Let $U^2_n$ be a cycle of length $2$ on $n\geq 4$ vertices such that one edge is of size $3$.

%
%
%

\begin{theorem} \label{Thm62}
Let $U$ be a unicyclic hypergraph on $n\geq 3$ vertices. Then $\rho_{\alpha}(U)\leq \rho_{\alpha}(U^1_n)$ with equality if and only if $U\cong U^1_n$. Moreover,
if $U\ncong U^1_n$ with $n\ge 4$, then $\rho_{\alpha}(U)\leq \rho_{\alpha}(U^2_n)$ with equality if and only if $U\cong U^2_n$.
\end{theorem}

\begin{proof}
Let $U$ be the unicyclic hypergraph on $n$ vertices with maximum $\alpha$-spectral radius.

Let $C=v_1 e_1 \ldots v_g e_g v_1$ be the unique cycle in $U$ with length $g$.
Suppose that $U\ne C$. Then there is an edge $e$ of $U$ not on $C$ and it contains some vertex in $V(C)$. Assume that $e$ contains some vertex $w$ of $e_1$. If $w\ne v_1, v_2$, then  by moving $e$ from $w$ to $v_1$ or moving $e_g$ from $v_1$ to $w$, we obtain a unicyclic hypergraph $U'$, and by Theorem \ref{Lem23}, we have $\rho_{\alpha}(U)<\rho_{\alpha}(U')$, a contradiction. Thus $w=v_1$ or $v_2$. Suppose that $w=v_1$. Let $U''$ be the unicyclic hypergraph obtained from $U$ by deleting edges $e$ and $e_1$ and adding an edge $e\cup e_1$.
%
Note that $A(U'')>A(U)$ and thus $A_{\alpha}(U'')>A_{\alpha}(U)$. Then by Lemma \ref{Lem22}, $\rho_{\alpha}(U'')>\rho_{\alpha}(U)$, a contradiction. Thus $U=C$. Assume that $|e_1|=\max\{|e_i|:1\le i\le g\}$.
If $g\geq 3$, then by setting $U^*$ to be the unicyclic hypergraph with girth $g-1\geq 2$ obtained from $U$ by deleting edges $e_2$ and $e_3$, and adding an edge $e_2\cup e_3$, we have by Lemma \ref{Lem22} that $\rho_{\alpha}(U^*)>\rho_{\alpha}(U)$, a contradiction. Thus $g=2$.
Then  $U=v_1 e_1 v_2 e_2 v_1$. By Theorem \ref{Thm33}, we have $|e_2|=2$, i.e., $U\cong U^1_n$.

If $U\ncong U^1_n$ with $n\ge 4$, then $\min\{|e_1|, |e_2|\}\geq 3$, and thus  by Theorem \ref{Thm33}, $\min\{|e_1|, |e_2|\}=3$, i.e., $U\cong U^2_n$.
\end{proof}

We note that the first part of the previous theorem follows also from Lemma \ref{Lem22}.

%
%

\begin{theorem} \label{Thm71}
Let $U$ be a unicyclic hypergraph on $n\geq 3$ vertices. Then $\rho_{\alpha}(U)\geq \rho_{\alpha}(C_n)$ with equality if and only if $U\cong C_n$.
\end{theorem}

\begin{proof}
Let $U$ be the unicyclic hypergraph on $n\geq3$ vertices with minimum $\alpha$-spectral radius.

It is trivial if $n=3$. Suppose that $n\geq 4$. If $U$ is a $2$-uniform unicyclic hypergraph, then $U\cong C_n$, since $C_n$ with $\rho_{\alpha}(C_n)=2$ is the unique unicyclic graph on $n$ vertices with minimum $\alpha$-spectral radius  among $2$-uniform unicyclic graphs. If $U$ is not a $2$-uniform unicyclic hypergraph, then there exists an edge $e_1=\{v_1, \ldots, v_r\}$ with size $r\geq 3$. By Lemma \ref{Lem22}, $\rho_{\alpha}(U)> \rho_{\alpha}(K_r)=r-1\geq 2$, a contradiction. Thus $U\cong C_n$.
\end{proof}

\begin{theorem} \label{Thm72}
Let $U$ be a unicyclic hypergraph on $n\geq 4$ vertices and $U\ncong C_n$. Then $\rho_{\alpha}(U)\geq \rho_{\alpha}(U_{n,n-1})$ with equality if and only if $U\cong U_{n,n-1}$.
\end{theorem}

\begin{proof}
Let $U\ncong C_n$ be a unicyclic hypergraph on $n\geq 4$ vertices with minimum $\alpha$-spectral radius.

Let $C$ be the unique cycle of $U$ with length $2\le g\le n-1$.

Suppose that $U$ is $2$-uniform. Let $C=v_1\dots v_gv_1$. By Lemma \ref{Lem24}, the deletion of edges on $C$ yields  $g$ vertex-disjoint paths $T_1$, \dots, $T_g$, where $v_i$ is a terminal vertex of $T_i$.  Suppose that  $g\le n-2$. Then  some  such path is nontrivial.  If only one such path, say $T_i$, is nontrivial, then there is an internal path  (of type (i)) from $v_i$ to $v_i$ in $U$.
If there are at least two such nontrivial paths, then we may choose two such paths, say $T_i$ and $T_j$  such that  $i<j$ implies that $T_{k}$  is trivial for each $k$ with  $i<k<j$, and thus there is an internal path (of type (ii))  from $v_i$ to $v_j$ in $U$. In either case, we may apply Theorem \ref{Lem26} for an edge, say $\{u,v\}$, on the internal path to form a unicyclic graph $U_{uv}$  such that $\rho_{\alpha} (U_{uv})<\rho_{\alpha} (U)$. 
Let $U'$ be the graph obtained from $U_{uv}$ by deleting a pendant vertex of the pendant path at $v_i$. Obviously, $U'$ is a unicyclic graph with $n$ vertices and girth $g+1$, and it is a proper subgraph of $U_{uv}$. By Lemma \ref{Lem22}, $\rho_{\alpha}(U')<\rho_{\alpha}(U_{uv})$. It follows that $\rho_{\alpha}(U')<\rho_{\alpha}(U)$, a contradiction. Thus $g=n-1$, and $U\cong U_{n,n-1}$.


Now suppose that $U$ is not $2$-uniform. Then there is an edge $e=\{v_1, \ldots, v_r\}$ with  $r\geq 3$. Suppose that $r\geq 4$. Since $U$ is unicyclic, there are two  vertices in $e$, say $v_1, v_2$, such that $\{v_1, v_2\}\notin E(U)$. Let $U'$ be the unicyclic hypergraph obtained from $U$ by deleting edge $e$ and adding edges $\{v_1, v_2\}$ and $\{v_2, \ldots, v_r\}$. Obviously, $U'$ is  not $2$-uniform. By Lemma \ref{Lem22}, $\rho_{\alpha}(U')<\rho_{\alpha}(U)$, a contradiction. Thus the largest size of edges is $3$. Suppose that there are two edges $e_1=\{u_1,u_2,u_3\}$ and $e_2$ with size $3$.
Let $U''$ be the unicyclic hypergraph obtained from $U$ by deleting edge $e_1$ and adding edges $\{u_1, u_2\}$ and $\{u_2, u_3\}$, where  if $e_1$  lies on  $C$, then we relabel the vertices and edges of $C$ as $u_1e_1u_3\dots u_1$.
Then  $U''$ is also not $2$-uniform, and by Lemma \ref{Lem22}, $\rho_{\alpha}(U'')<\rho_{\alpha}(U)$, a contradiction. Thus there is exactly one edge, say $e^*=\{w_1, w_2, w_3\}$ with size $3$. If $e^*$ is not an edge on $C$ or $e^*$ is an edge on $C$ with $g\le n-2$, then let $U'''$ be the unicyclic hypergraph obtained from $U$ by deleting edge $e^*$ and adding edges $\{w_1, w_2\}$ and $\{w_2, w_3\}$, where
if in the latter case, the vertices and edges of $C$ are relabeled as $w_1e^*w_3\dots w_1$.
Then $U'''$ is  $2$-uniform and $U'''\ncong C_n$. By Lemma \ref{Lem22}, $\rho_{\alpha}(U''')<\rho_{\alpha}(U)$, a contradiction. Thus $e^*$ is an edge on the cycle of $U$  and $g=n-1$, i.e., $U\cong C$.

Let $U^*$ be the unicyclic graph obtained from $U$ by removing $w_2$ from $e^*$ and attaching an edge $\{w_1,w_2\}$ to $w_1$. Then $U^*\cong U_{n,n-1}$. By Theorem \ref{Thm34}, $\rho_{\alpha}(U^*)<\rho_{\alpha}(U)$, a contradiction.

Thus $U$ is $2$-uniform, and $U\cong U_{n,n-1}$.
\end{proof}

\section{Maximum $\alpha$-spectral radius of hypergraphs with fixed number of pendant edges}

For $2\le k\le n$, the complete $k$-uniform hypergraph, denoted by $K_n^{(k)}$, is a hypergraph $G$
of order $n$ such that $E(G)$ consists of all $k$-subsets of $V(G)$.


For $2\leq k\leq n$ and $p\geq 0$, let $\mathbb{F}_{n,k}(p)$ be the set of connected $k$-uniform hypergraphs of order $n$ with $p$ pendant edges.
Since $G\in \mathbb{F}_{n,k}(p)$ is connected, we have $n-p(k-1)=1$ with $p>0$ and then $G\cong S_{n,k}$, or $n-p(k-1)\geq k$. 

For $k\geq 2$ and $p\geq 0$, let $K^{(k)}_{n,p}$ be the $k$-uniform hypergraph obtained from $K^{(k)}_{n-pk+p}$ by attaching $p$ pendant edges at a vertex of $K^{(k)}_{n-pk+p}$.

\begin{theorem}\label{hypertree-adj-pendant}
For $2\leq k \leq n$, let $G\in \mathbb{F}_{n,k}(p)$. Then

(i) if $p=0$, then $\rho_{\alpha}(G)\le {n-2\choose k-2}(n-1)$ with equality if and only if $G\cong K^{(k)}_n$;

(ii) if $p\geq  1$, then $\rho_{\alpha}(G)\leq \rho_{\alpha}(D_{n,k,1})$ with equality if and only if $G\cong D_{n,k,1}$ if $n-p(k-1)=k$ with $p\ge 2$, and
$\rho_{\alpha}(G)\leq \rho_{\alpha}\left(K^{(k)}_{n,p}\right)$ with equality if and only if $G\cong K^{(k)}_{n,p}$ if $n-p(k-1)>k$.
\end{theorem}

\begin{proof}
Let $G$ be a hypergraph in $\mathbb{F}_{n,k}(p)$ with maximum $\alpha$-spectral radius.

Let $V_0$ be the set of vertices that are not vertices of degree $1$ in the $p$ pendant edges of $G$. By Lemma \ref{Lem22}, $G[V_0]$ is a complete $k$-uniform hypergraph.

Obviously, if $p=0$, then  $G=G[V_0]\cong K^{(k)}_n$, and then $\rho_{\alpha}(G)={n-2\choose k-2}(n-1)$. This is (i).

Suppose that $p\geq1$.

If $n-p(k-1)=k$, then $G[V_0]=K^{(k)}_k$ contains a single edge, which is the only non-pendant edge of $G$, and the $p$ pendant edges are attached at at least two vertices of this edge. Obviously, $G$ is a $k$-uniform hypertree with $p\ge 2$. By Theorem \ref{NSmax},  $G\cong D_{n,k,1}$.

If $n-p(k-1)>k$. Suppose that there are pendant edges at  different vertices $u$ and $v$ of $G$.  We may assume that   $x_{u}\geq x_{v}$, where $x=x(G)$.
Let $G'$ be the $k$-uniform hypergraph obtained from $G$ by moving all pendant edges at $v$ from $v$ to $u$. Obviously, $G'\in \mathbb{F}_{n,k}(p)$.
By Theorem \ref{Lem23}, $\rho_{\alpha}(G')>\rho_{\alpha}(G)$, a contradiction.
Thus all pendant edges share a common vertex, and
 $G\cong K^{(k)}_{n, p}$. This proves (ii).
\end{proof}

If $G$ is a hypergraph on $n$ vertices with $n-1$ pendant edges, then $G\cong S_n$.

Let $G(n)$ be the hypergraph on $n$ vertices  with all possible edges, that is, any vertex subset of cardinality at least two is an edge.

For $0\le p\le n-3$, let $H^p_n$ be the hypergraph on $n\geq p+3$ vertices obtained from $G(n-p)$ by attaching $p$ pendant edges of size two to one vertex.

By Corollary \ref{Cor42}, $S^p_n$ is the unique hypertree on $n$ vertices with $p$ pendant edges having maximum $\alpha$-spectral radius. To contrast with this, we have

\begin{theorem} \label{Thm81}
Let $G$ be a hypergraph on $n$ vertices with $p$ pendant edges having maximum $\alpha$-spectral radius, where $0\leq p \leq n-2$.

(i) if $p=n-2$, then $G\cong S^{n-2}_n$;

(ii) if $0\le p\le n-3$, then $G\cong H^p_n$.
\end{theorem}

\begin{proof}
Suppose that $p=n-2$. If $G$ is $2$-uniform, then $G$ is isomorphic to some double star $S(n_1,n_2)$ obtained by adding an edge between the centers of two nontrivial stars $S_{n_1}$ and $S_{n_2}$ with $n_1+n_2=n$, and by Theorem \ref{Lem23},   $\min\{n_1, n_2\}=2$, i.e., $G\cong S(n-2,2)$.  If $G$ is not $2$-uniform, then $G\cong S^{n-2}_n$. By Lemma \ref{Lem22} or Theorem \ref{Thm34}, $\rho_{\alpha}(S^{n-2}_n) > \rho_{\alpha}(S(n-2,2))$. Thus $G\cong S^{n-2}_n$.

Next suppose that $0\le p\le n-3$.
It is trivial if $p=0$.
Suppose that $1\leq p \leq n-3$. 

Let  $e_1, \dots, e_p$ be the $p$ pendant edges and $V_1$ be the set of vertices in pendant edges with degree $1$. Let $V_0=V\setminus V_1$. By Lemma \ref{Lem22},  any subset of $V_0$ with at least two vertices is an edge of $G$.

Suppose that  $|V_0|= 2$. Let $V_0=\{u,v\}$. Assume that $x_u\ge x_v$.  Let  $G'$ be the hypergraph obtained from $G$ by moving pendant edges at $v$ from $v$  to $u$. By Theorem \ref{Lem23}, $\rho_{\alpha}(G)<\rho_{\alpha}(G')$. Note that $G'$ is a hypergraph on $n$ vertices with $p+1$ pendant edges, one of which is $\{u,v\}$. Since $p\le n-3$, there is a pendant edge $e_1$ with size at least $3$. Let  $G''$ be the hypergraph obtained from $G'$ by adding all possible edges $e$ such that $e\subseteq e_1$ with $|e|\ge 2$. Obviously, $G''$ is a hypergraph on $n$ vertices with $p$ pendant edges. By Lemma \ref{Lem22}, $\rho_{\alpha}(G')<\rho_{\alpha}(G'')$. It follows that $\rho_{\alpha}(G)<\rho_{\alpha}(G'')$, a contradiction. Therefore $|V_0|\ge 3$ or $|V_0|= 1$. By Theorem \ref{Lem23}, the $p$ pendant edges are at a common vertex, say $v_1\in V_0$. By Theorem \ref{Thm31}, there is at most one pendant edge with size at least $3$. Thus we may assume  that $|e_1|\geq |e_2|= \dots =|e_p|=2$.
Let $e_1=\{v_1, \dots, v_r\}$. Suppose that $r\geq3$. Let $G^*$ be the hypergraph obtained from $G$ by deleting edge $e_1$, adding edge $\{v_1,v_2\}$ and all possible edges $e$ containing at least one vertex in $\{v_3, \dots, v_r\}$ such that
$e\subseteq \{v_3, \dots, v_r\}\cup V_0$.
  Obviously, $G^*$ is also a hypergraph on $n$ vertices with $p$ pendant edges. For $\{w,z\}\subseteq V(G)$, it is easily seen that
\begin{align*}
&\quad a_{wz}(G^*)-a_{wz}(G)\\
&=\begin{cases}
2^{n-p-2}-1 & \mbox{if } w=v_1, z\in \{v_3, \dots, v_r\}, \mbox{ or }w\in \{v_3, \dots, v_r\}, z=v_1,\\
-1 & \mbox{if } w=v_2, z\in \{v_3, \dots, v_r\}, \mbox{ or }w\in \{v_3, \dots, v_r\}, z=v_2,
\end{cases}
\end{align*}
and $a_{wz}(G^*)-a_{wz}(G)\geq 0$ otherwise. Let $x=x(G)$. By Lemma \ref{Lem21}, $x_{v_2}=x_{v_3}= \dots = x_{v_r}$ and $x_u=x_v$ for $u,v\in V_0 \setminus \{v_1\}$. Note that $A_{\alpha}(G[e_1])=A_{\alpha}(K_{r})$ is a principal submatrix of $A_{\alpha}(G)$. By Lemma \ref{Lem22}, we have $\rho_{\alpha}(G)>\rho_{\alpha}(K_r)=r-1$. From the eigenequation of $G$ at $v_2$, we have
\[
\rho_{\alpha}(G)x_{v_2}=\alpha(r-1)x_{v_2}+(1-\alpha)\left(x_{v_1}+(r-2)x_{v_2}\right),
\]
which implies  $(1-\alpha)x_{v_1}=(\rho_{\alpha}(G)-r+2-\alpha)x_{v_2}$, and thus $x_{v_1}>x_{v_2}$. Therefore
\begin{align*}
&\quad \rho_{\alpha}(G^*)-\rho_{\alpha}(G)\\&\geq x^\top (A_{\alpha}(G^*)-A_{\alpha}(G)) x  \\
&= \sum_{\{w,z\}\subseteq V(G)} (a_{wz}(G^*)-a_{wz}(G))\left(\alpha(x_w^2+x_z^2)+2(1-\alpha)x_w x_z\right) \\
&= \sum_{z\in \{v_3, \dots, v_r\}}(2^{n-p-2}-1)(\alpha(x_{v_1}^2+x_z^2)+2(1-\alpha)x_{v_1} x_z)\\
&\quad -\sum_{z\in \{v_3, \dots, v_r\}}(\alpha(x_{v_{2}}^2+x_z^2)+2(1-\alpha)x_{v_2} x_z)\\
&=\alpha(r-2)(2^{n-p-2}-1)x_{v_1}^2+\alpha(r-2)(2^{n-p-2}-1)x_{v_2}^2\\
&\quad -2\alpha(r-2)x_{v_2}^2+2(1-\alpha)\sum_{z\in\{v_3,\dots,v_r\}}x_z(x_{v_1}-x_{v_2})\\
&=\alpha(r-2)(2^{n-p-2}-1)x_{v_1}^2+(\alpha(r-2)(2^{n-p-2}-3)-2(1-\alpha)(r-2))x_{v_2}^2\\
&\quad +2(1-\alpha)(r-2)x_{v_1}x_{v_2}\\
&>\left(\alpha(r-2)(2^{n-p-2}-1)+\alpha(r-2)(2^{n-p-2}-3)-2(1-\alpha)(r-2)\right.\\
&\quad \left.+2(1-\alpha)(r-2)\right)x_{v_2}^2\\
&=2\alpha(2^{n-p-2}-2)(r-2)x_{v_2}^2\\
&\geq  0,
\end{align*}
implying $\rho_{\alpha}(G^*)>\rho_{\alpha}(G)$, a contradiction.  Therefore $|e_1|= \dots =|e_p|=2$, and $G\cong H^p_n$.
\end{proof}

\section{Concluding remarks}

The $0$-spectral radius of a uniform hypergraph has been studied in \cite{LZ}, while it seems that there is no study on the $0$-spectral radius of a  hypergraph that is not necessarily uniform (via the above adjacency matrix).
In this paper, we  consider the $\alpha$-spectral radius of hypergraphs that are uniform or not necessarily uniform.
We propose  some local grafting operations that increase or decrease  the $\alpha$-spectral radius. Among others, we identify
the unique $k$-uniform hypertrees with the first three largest  $\alpha$-spectral radii, the unique $k$-uniform unicyclic hypergraphs with maximum $\alpha$-spectral radius,
the unique  $k$-uniform hypergraphs with  maximum $\alpha$-spectral radius when the number of pendant edges is given, and we also
identify the unique hypertrees with  maximum $\alpha$-spectral radius among hypertrees with given number of vertices and edges, the unique hypertrees with the first three largest (two smallest, respectively) $\alpha$-spectral radii among hypertrees with given number of vertices, the unique hypertrees with minimum $\alpha$-spectral radius among the hypertrees  that are not $2$-uniform with given number of vertices, the unique hypergraphs with the first two largest (smallest, respectively) $\alpha$-spectral radii among unicyclic hypergraphs with given number of vertices, and  the unique hypergraphs with maximum $\alpha$-spectral radius among hypergraphs with fixed number of pendant edges.

\vspace{10mm}
	
\noindent {\bf Acknowledgement.} This work was supported by the National Natural Science Foundation of China (Nos. 12071158 and 11671156).

\end{document}